\documentclass[12pt]{amsart}
\usepackage{amsmath,amssymb,amscd}
\usepackage{graphicx}
\usepackage{braket}

\makeatletter
\@addtoreset{equation}{section}
\makeatother
\usepackage{enumerate}

\usepackage{color}

\def\CC{{\mathbb C}}

\def\KK{{\mathbb K}}

\def\ZZ{{\mathbb Z}}
\def\QQ{{\mathbb Q}}

\def\ii{{\mathfrak{i}}}

\def\cU{{\mathcal U}}

\def\frki{{\mathfrak i}}

\def\rU{{\mathrm U}}

\def\hPsi{{\widehat \Psi}}
\def\hPhi{{\widehat \Phi}}

\def\det{{\mathrm {det}}}

\def\hcnj{{\, {}^*\negthinspace}}

\def\book#1{\rm{#1}, }
\def\paper#1{\textit{#1}, }
\def\jour#1{\rm{#1}, }
\def\yr#1{({\rm{#1}) }}
\def\vol#1{\textbf{#1}}
\def\pages#1{\rm{#1}}

\def\publaddr#1{\rm{#1}, }
\def\publ#1{\rm{#1}, }
\def\by#1{{\rm{#1}, }}

\usepackage[utf8]{inputenc}

\newtheorem{definition}{Definition}[section]

\newtheorem{theorem}[definition]{Theorem}
\newtheorem{proposition}[definition]{Proposition}
\newtheorem{corollary}[definition]{Corollary}
\newtheorem{remark}[definition]{Remark}

\def\dfrac#1#2{{\displaystyle\frac{#1}{#2}}}

\def\book#1{\rm{#1}, }
\def\paper#1{\textit{#1}, }
\def\jour#1{\rm{#1}, }
\def\yr#1{({\rm{#1}) }}

\def\vol#1{\textbf{#1}}
\def\pages#1{\rm{#1}}

\def\publaddr#1{\rm{#1}, }
\def\publ#1{\rm{#1}, }
\def\by#1{{\rm{#1}, }}

\begin{document}

\title{Wave-particle Complementarity in Quantum Walk from a quantum mechanical perspective}

\author{Norio Konno}
\address{College of Science and Engineering, Ritsumeikan University, 1-1-1 Noji-higashi, Kusatsu, 525-8577, Shiga, Japan}
\email{n-konno@fc.ritsumei.ac.jp}

\author{Shigeki Matsutani}
\address{Graduate School of Natural Science \& Technology, Kanazawa University, Kakuma, Kanazawa, 920-1192, Ishikawa, Japan}
\email{s-matsutani@se.kanazawa-u.ac.jp}

\author{Etsuo Segawa}
\address{Graduate School of Environment and Information Sciences, Yokohama National University, 79-1 Tokiwadai, Hodogaya, Yokohama, 240-8501, Kanagawa, Japan}
\email{segawa-etsuo-tb@ynu.ac.jp}

\author{Tomoyuki Terada}
\address{Corporate Management Division, Kanazawa Institute of Technology, 7-1 Ohgigaoka, Nonoichi, 921-8501, Ishikawa, Japan}
\email{t.terada@neptune.kanazawa-it.ac.jp}

\date{\today}


\begin{abstract}
A quantum walk is a mathematical model and a quantum version of a random walk. For the past two decades, it has been studied from mathematical and physical viewpoints.
It has many desirable properties.
This paper reevaluates quantum walk from a quantum mechanical perspective, as Feynman essentially proposed it as a toy model in his path integral method.
Then we demonstrate that the quantum walk model exhibits both particle-like and wave-like properties by showing reflection and partial reflection as a particle and interference as a wave.
In particular, for an initial state supported at two sites $\pm\ell$, we explicitly derive the weak limit distribution of the Hadamard walk and show that the Konno distribution is modulated by the interference factor $2\cos^2(\ell k(x))$. The resulting interference profile is reminiscent of far-field (Fraunhofer-type) diffraction.
\end{abstract}
\maketitle

\section{Introduction}
It is well known that Feynman proposed a mathematical model -- the so-called discrete Dirac equation model -- to help understand the path integral picture in quantum mechanics \cite[p.35, Problem 2-6]{FHS} (See Appendix A).
This model is essentially a prototype of the quantum walk (QW), which is currently being studied in various fields and is the main theme of this paper.
Feynman himself introduced the concept of the path integral in a lecture aimed at a general audience in a more plain way, which was later published as a book \cite{F}. 
There, he also demonstrated that particles described by path integrals obey Fermat's principle via thought experiments.

The QW itself was introduced by Aharonov, Davidovich, and Zagury \cite{ADZ}.

Recently due to the various excellent studies of QW, its mathematical properties and characteristics have been revealed.
The reviews for developments of the theory of QWs are, for example, Kempe \cite{Kem}, Kendon \cite{Ken}, Venegas-Andraca \cite{VAndraca2008, VA}, Konno \cite{Konno2008b, Kon14},  Manouchehri and Wang \cite{ManouchehriWang2013}, Portugal \cite{Portugal2013}. 
Konno analytically showed that the homogeneous scattering by Hadamard type matrix at each point forms the {\lq\lq}wavefront{\rq\rq} in QW, which is known as the Konno distribution \cite{Kon05,Konno2008b}, as mentioned in Theorem \ref{th:Konno}. 
Further, there are so many mathematical investigations on QWs \cite{KS13, GJS}.

\bigskip

On the other hand, complementary properties have been studied in quantum mechanics for a long period of time.
Although there were many arguments about quantum mechanics in the early days, a few models and computations demonstrate both properties using a parameter, e.g. the Planck constant $\hbar$.

However, we found that the QW has several desirable properties as a quantum mechanical model, in addition to being a mathematical and probability model.
This paper demonstrates that in the QW, the particle-like and wave-like properties are inherent.
In other words, the QW may provide a rigorous mathematical foundation for the arguments concerning wave-particle complementarity and fundamental concepts in the quantum mechanics.

In fact, Gudder essentially found an algorithm for the QW on the real axis in order to understand the fundamental properties of quantum mechanics \cite{Gud, GudN}.
However, it was difficult to handle such a model precisely in those days.

Therefore, we should return to re-evaluate the physical properties of the QW as a toy model of quantum mechanics or the path integral, as Feynman partially demonstrated (see Appendix).
We refer to it as the QW model from the quantum mechanical perspective.

Very recently, Fouque, Ichiba and Lam \cite{FIL} reported on research concerning path integrals based on quantum walks (QWs), which was motivated by reasons similar to those in this paper.
While there is no overlap with this paper in terms of research content, their approach --which involves linking fundamental problems in path integrals to QWs-- is extremely interesting.

\bigskip

The contents of this paper are as follows.

Section 2 provides an overview on the QW including the Konno distribution in Theorem \ref{th:Konno} and the review of a method proposed by Grimmett, Janson, and Scudo \cite{GJS, Kon14}.
Section 3 demonstrates that the QW provides single-particle mechanics by considering the peak of the wave profile in numerical QW computations. 
In the QW model, we refer to this peak as the wavefront.
In the QW, we observe a wavefront that behaves like a particle.
By demonstrating the reflection and partial reflection of the wavefront of the QW model as a particle, we show that the quantum state exhibits particle-like behavior, which is reminiscent of the results of Feynman's thought experiments for  Fermat's principle \cite{F}.
Section 4 demonstrates that QWs provide an interference phenomenon by considering quantum states stemming from two points.
Though it reminds us of the double-slit phenomenon of the far-field, it might be related to quantization of the dipole in the near-field type interference like the Talbot effect \cite{BK, MO} because it is characterized by a natural number.
Either way, we observe the interference.
In Section 5, we give the conclusion in this paper and the discussion.
Appendix A gives a short review of Feynman's discrete Dirac equation model, and Appendix B shows the proof of Proposition 2.4.

\section{Quantum Walk: Hadamard walk}

Let us consider the discrete space-time $(x, t) \in \ZZ^2$.
In this paper, $\ZZ$ denotes the set of integers and $\CC$ denotes the set of complex numbers.

The time evolution of the QW is given by the infinite unitary matrix $\cU \in \rU(\infty)$ determined by the sequence of 2-dimensional unitary matrices $(U_x)_{x\in\ZZ}$,
$$
  U_x = \left[\begin{matrix} a_x & b_x \\ c_x & d_x\end{matrix}\right],
$$
where $a_x$, $b_x$, $c_x$, $d_x \in  \CC$ satisfy
$$
|a_x|^2 +|c_x|^2=|b_x|^2+|d_x|^2=1, \quad
a_x\overline{b_x}+c_x\overline{d_x}=0, \quad
c_x=-\det(U_x) \overline{b_x}, \quad
d_x=\det(U_x) \overline{a_x}
$$
due to the unitarity of $U_x$.
Further we set
\begin{equation}
     P_x :=  \left[\begin{matrix} a_x & b_x \\ 0 & 0\end{matrix}\right], \quad
Q_x :=\left[\begin{matrix} 0 & 0 \\ c_x & d_x\end{matrix}\right].
\label{eq:P,Q}
\end{equation}
In the quantum walk literature, $U_x$ is called the quantum coin.
In this paper, we also refer to the distribution $\cU$ of $U_x$ as quantum circumstance.
The $\cU$ acts on the vector
\begin{equation}
\Psi(t):=\left\{\Psi(t,x):=\left[\begin{matrix} \Psi_L(t,x) \\ \Psi_R(t,x)
 \end{matrix}\right] \ \Bigr| \ x \in \ZZ\right\}\in (\CC^2)^{\ZZ} \cong \CC^{\ZZ\times\{L,R\}}
\label{eq:Psi}
\end{equation}
for $t \in \ZZ_{\ge 0}:=\{t \in \ZZ \ |\ t \ge 0\}$ as the recursion relation,
\begin{equation}
    \Psi(t+1) = \cU\Psi(t)
\label{eq:UPsi}
\end{equation}
satisfying the following relation,
\begin{equation}
      \Psi(t+1, x) = P_{x+1}  \Psi(t, x+1)+Q_{x-1}  \Psi(t, x-1).
\label{eq:PQPsi}
\end{equation}
Here $\Psi_j(t,x)$ is a complex number for $j=L, R$ or $\Psi(t, x) \in \CC^2$.
We regard them as $\CC^2$-valued functions $\Psi(\cdot, \cdot): \ZZ\times \ZZ \to \CC^2$ on $\ZZ^2$ or $\Psi(t, \cdot): \ZZ \to \CC^2$ for $t \in \ZZ$, or $\Psi_j(t, \cdot ): \ZZ \to \CC$ on $\ZZ$ for $t\in \ZZ$ and $j \in \{L, R\}$.
This recursion relation (\ref{eq:UPsi}) is referred to as the QW of Ambainis-type \cite{ABNVW}.
We call $\Psi$ the quantum state in this paper.

The system of the QW has two Hermitian inner products, the point-wise, 
$(\Phi(t, x),\Psi(t, x))_{x,t}=\Phi \hcnj(t, x)\Psi(t, x)$ and the global one,
$(\Phi(t), \Psi(t))_t=\Phi\hcnj(t) \Psi(t)=\displaystyle{\sum_{x\in \ZZ}\Phi \hcnj(t, x)\Psi(t, x)}$.
In the QW, the probability distribution is given by the point-wise Hermitian inner product,
$$
      \rho(t,x) =(\Psi(t, x),\Psi(t, x))_{x,t}:=\Psi \hcnj(t, x)\Psi(t, x).
$$
Obviously, if $\displaystyle{\sum_{x\in \ZZ} \rho(0,x)=1}$, we have $\displaystyle{\sum_{x\in \ZZ} \rho(t,x)=1}$ for $t \in \ZZ$, and the positivity $\rho(t,x)\ge 0$.

\bigskip

This picture is closely related to Feynman path integral picture mentioned in the following remark.

\begin{remark}\label{rmk:PIvsQW}
{\rm{
In the QW, the state of $t+1$ at point $x$ is determined by the states of $t$ in the nearest neighborhood of the point $x$, i.e., $x-1$ and $x+1$.
This seems to differ from Feynman's path integral approach.

However, we can interpret the time interval $\Delta t$ as $N$ time steps for the same algorithm of the QW model, i.e., $\Delta t = N$ [time step].
The state at point $x$ at time $t + \Delta t$ is determined by the states at time $t$ within the neighborhood of size $2 \Delta t+1$ around point $x$, i.e., from $x - \Delta t$ to $x + \Delta t$.
In other words, we can describe it as
\begin{equation}
  \Psi(t_0 + \ell \Delta t,x) 
= \sum_{n=-N}^{N} K_{\Delta t}(x, x+n) \Psi(t_0 + (\ell-1) \Delta t, x+n),
\label{eq:PsiKPsi}
\end{equation}
where $K_{\Delta t}(x, y)$ is the kernel function given by $\cU$-action (\ref{eq:UPsi}).
Here (\ref{eq:PsiKPsi}) implicitly appeared in the generating function method in QW \cite{KS13}.

If we take $\Delta t=\infty$, it would correspond to Feynman's picture of the path integral even though we have to investigate whether $K_{\Delta t}$ is connected with a physical model or the Hamiltonian.
This means that the QW could be interpreted as a mathematical model of Feynman's path integral over a finite range. 
At the very least, it has the potential to be a mathematical model.
}}
\end{remark}

\bigskip

\begin{figure}
\begin{center}

\includegraphics[width=0.7\hsize]{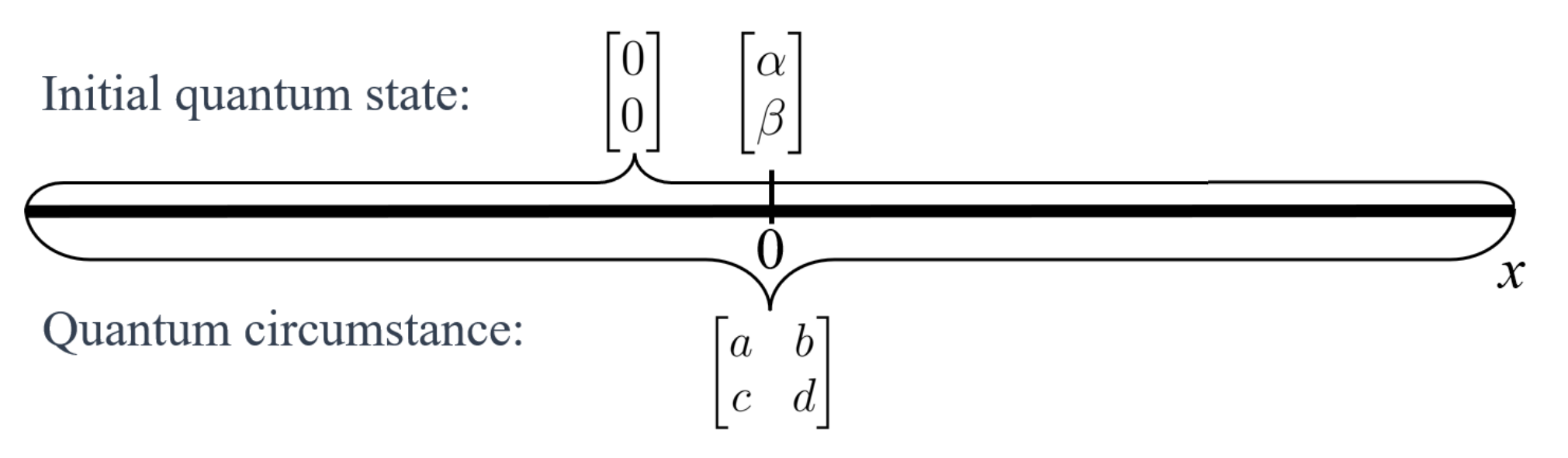}
\

\smallskip

\end{center}

\caption{
Constant quantum circumstance.
}\label{fg:system_Kdist}
\end{figure}

Konno mathematically determined the envelope curve of the  distribution $\rho(t,x)$ \cite{Kon05}:
\begin{theorem}[\cite{Kon05}]\label{th:Konno}
Suppose $abcd \neq 0$.
For the initial state
$\Psi(0,x)= \varphi \delta_{x,0}$, 
$\varphi = \left[\begin{matrix} \alpha \\ \beta \end{matrix}\right]\in\mathbb{C}^{2}$ with $|\alpha|^2+|\beta|^2=1$, and the quantum circumstance (quantum coin) as in Figure \ref{fg:system_Kdist},
$$
U_x= U= \left[\begin{matrix}
a & b \\ c & d\end{matrix}\right], \quad x \in \ZZ.
$$
Let $X_t$ be a random variable following $\rho(t,\cdot)$ for $t \in \ZZ_{\ge 0}$.
Then we have the probability of the quantum state (q-bits) as
$$
\lim_{t\to \infty}
P\left(u \le  \frac{X_t}{t} \le v\right)
=\int^v_u f_{K}(x) dx \quad \text{for any } u \le v,
$$
where
$$
f_K(x) = \frac{\sqrt{1-|a|^2}}
  {\pi (1-x^2)\sqrt{|a|^2 -x ^2}}
I_{(-|a|,|a|)}(x) \times
      \left[ 1 - \left(
  |\alpha|^2 -|\beta|^2 +
\frac{a\alpha \overline{b\beta} + \overline{a\alpha}b\beta}
{|a|^2}\right) x \right].
$$
\end{theorem}
\bigskip

\begin{proof}
See \cite{Kon05}.
Even though Konno proved this by means of a combinatorial method, we show its proof of a special case following \cite{Kon05} as Proposition \ref{pr:KonnoDist} and Appendix B.
\end{proof}

Especially, we are concerned with the Hadamard walk such that
\begin{equation}
 U_x = \frac{1}{\sqrt{2}}\left[\begin{matrix} 1 & 1 \\ 1 & -1\end{matrix}\right].
\label{eq:Harmard}
\end{equation}

\subsection{Introduction to the GJS method}

The method proposed by Grimmett, Janson and Scudo provides an elegant analytic framework for deriving the weak limit distribution of discrete-time QWs \cite{GJS}.
We call it the GJS-method.

Since the GJS method plays a crucial role in Proposition \ref{pr:KonnoDist} and Theorem \ref{thm:Wslit}, in this section we outline the general argument and apply it explicitly to the Hadamard walk, following the spectral formulation.

Let $\Psi(t,x) = \displaystyle{\left[\begin{matrix}\Psi_{L}(t,x)\\ \Psi_{R}(t,x) \end{matrix}\right]} \in \mathbb{C}^2$ denote the complex vector as the state  at site $x\in\ZZ$ and time $t \in\ZZ$,
where the subscripts $L$ and $R$ correspond to the left and right states in the QW-model, respectively.

Let $\KK=[-\pi,\pi)$.
For $k \in \KK$, we consider the Fourier transform of each component by
$$
  \widehat{\Psi}_{j}(t,k)
  = \sum_{x\in\mathbb{Z}} e^{-ikx}\Psi_{j}(t,x),
  \qquad (j=L, R)
$$
with the inverse relation,
$$
  \Psi_{j}(t,x)
  = \int_{-\pi}^{\pi} e^{ikx}\widehat{\Psi}_{j}(t,k) \frac{dk}{2\pi}.
$$
Then for $k \in \KK$, the Fourier-space state can be expressed as
$$
  \widehat{\Psi}(t,k)
  =
  \begin{bmatrix}
    \widehat{\Psi}_{L}(t,k) \\[2pt]
    \widehat{\Psi}_{R}(t,k)
  \end{bmatrix}
  \in \mathbb{C}^2,
$$
and thus the Fourier-space evolution takes the form
$$
  \widehat{\Psi}(t+1,k) = U(k)\widehat{\Psi}(t,k),
$$
where $U(k)\in \rU(2)$ denotes the Fourier representation of the one-step evolution operator.
The Fourier-space also admits two Hermitian inner products, the point-wise 
$(\hPhi(t, k),\hPsi(t, k))_{k,t}=\hPhi\hcnj(t, k)\hPsi(t, k)$ and the global one,
$(\hPhi(t), \hPsi(t))_t=\displaystyle{ \int_{-\pi}^{\pi}  \hPhi \hcnj(t, k)\hPsi(t, k) \frac{dk}{2\pi}}$.
Further, for later convenience, we also prepare the Hermitian point-wise inner product $(u(k), v(k))_k = u(k)^\hcnj v(k)$ for $u(k) \in \CC^2$.
Let $\lambda_j(k)$ and $v_{j}(k)$ $(j=1,2)$ be the eigenvalues and the corresponding eigenvectors of $U(k)$, chosen so that $\{v_1(k), v_2(k)\}$ is an orthonormal basis of $\CC^2$ for each $k \in \KK$:
$$
      U(k) v_j(k) = \lambda_j(k) v_j(k), \quad j = 1,2.
$$
In other words, we write the spectral decomposition of $U(k)$ as 
$$
U(k)=\sum_{j=1}^{2}\lambda_j(k) v_j(k)\otimes_\CC v_j^*(k).
$$
Then the Fourier-space amplitude can be expressed as
\begin{align*}
\widehat{\Psi}(t,k) &= U(k)^{t} \widehat{\Psi}(0,k)  \\
&=\lambda_{1}(k)^{t} ( v_{1}(k) ,  \widehat{\Psi}(0,k))_{k,0}   v_{1}(k)
 + \lambda_{2}(k)^{t} ( v_{2}(k) ,  \widehat{\Psi}(0,k))_{k,0} v_{2}(k),
\end{align*}
where the initial state (initial coin) is $\widehat{\Psi}(0,k)=\displaystyle{\left[
\begin{matrix}\alpha\\ \beta\end{matrix}\right]}$ with $|\alpha|^2+|\beta|^2=1$.

Let $D:=i\,\dfrac{d}{dk}$. For the position random variable $X_{t}\in\mathbb{Z}$, its $r$-th moment admits the standard representation
$$
E(X_{t}^{r}) = \int_{-\pi}^{\pi} \frac{dk}{2\pi} \widehat{\Psi}^{\ast}(t,k) D^{r} \widehat{\Psi}(t,k).
$$
The GJS--method expresses the large-$t$ asymptotics of $E(X_{t}^{r})$.
The key quantity is the group velocity
\[
h_j(k)=\frac{D\lambda_j(k)}{\lambda_j(k)} \qquad (j=1,2).
\]
Repeated differentiation of the spectral decomposition gives
$$
D^{r} \widehat{\Psi}(t,k) = \sum_{j=1}^{2} (t)_{r}  \lambda_{j}(k)^{t-r} \bigl( D  \lambda_{j}(k) \bigr)^{r} (v_{j}(k),  \widehat{\Psi}(0,k))_{k,0} v_{j}(k)  + O(t^{r-1}),
$$
where $(t)_{r}=t(t-1)\cdots(t-r+1)$.
Substituting this into the moment expression yields
$$
E \left( \left( \frac{X_{t}}{t} \right)^{r} \right) = \int_{-\pi}^{\pi} \sum_{j=1}^{2} \frac{(t)_{r}}{t^r} \left( \frac{D\lambda_{j}(k)}{\lambda_{j}(k)} \right)^{r} \lvert( v_{j}(k), \widehat{\Psi}(0,k))_{k,0}\rvert^2 \frac{dk}{2\pi} + O(t^{-1}).
$$

\begin{theorem}\label{th:GJS}
Let $X_{t}$ denote the QW at time $t$.
We define $h: \Omega = \mathbb{K} \times \{1,2 \} \to \mathbb{R}$ by
\[
h(k,j) = h_{j}(k) = \frac{D\lambda_{j}(k)}{\lambda_{j}(k)}.
\]
Then, we have
$$
\frac{X_{t}}{t} \Rightarrow Y = h(Z) \quad (t \to \infty),
$$

where $Z$ is a random variable on $\Omega$ whose distribution is given by
$\bigl|(v_j(k), \widehat{\Psi}(0,k))_{k,0}\bigr|^2\, \dfrac{dk}{2\pi}$ on $\KK \times \{1,2\}$.

\end{theorem}

\begin{proof}
The equation above shows convergence of all moments of $\dfrac{X_{t}}{t}$ to those of $h(Z)$. Since $|h_{j}(k)|\le 1$, the distribution of $h(Z)$ has compact support and is therefore uniquely determined by its moments. Hence, the method of moments yields weak convergence.
\end{proof}

\begin{proposition}\label{pr:KonnoDist}
We let $U(k)\in \rU(2)$ be the Fourier representation of the Hadamard case (\ref{eq:Harmard}),
\[
U(k)=\frac{1}{\sqrt{2}}
\begin{bmatrix}
e^{ik}   &e^{ik} \\
e^{-ik}  &-e^{-ik}
\end{bmatrix},
\]
and the initial state be $\widehat{\Psi}(0,k)=\displaystyle{
\left[\begin{matrix}\alpha \\ \beta\end{matrix}\right]}$
with $|\alpha|^{2}+|\beta|^{2}=1$.
Further we let $X_t\in\mathbb{Z}$ denote the position of the quantum walker (the quantum state arising from $(t,x) = (0,0)$) after $t$ steps.
Then the rescaled position variable $X_t/t$ converges in distribution to a random variable $Y$ supported on $(-1/\sqrt{2},\,1/\sqrt{2})$, whose probability density function is
\begin{align*}
f(y) &= \Bigl\{1-\bigl(|\alpha|^{2}-|\beta|^{2}+\alpha\overline{\beta}+\overline{\alpha}\beta\bigr)y \Bigr\}
\frac{1}{\pi (1-y^2) \sqrt{1-2y^2}} I_{(-1/\sqrt{2},\,1/\sqrt{2})}(y) \\
&= \Bigl\{1-\bigl(|\alpha|^{2}-|\beta|^{2}+\alpha\overline{\beta}+\overline{\alpha}\beta\bigr)y \Bigr\} f_{K}\!\left(y;\,\frac{1}{\sqrt{2}}\right).
\end{align*}
In particular, $Y$ has the Konno distribution $f_{K}(y;1/\sqrt{2})$ modulated by a linear tilt depending on the initial state $(\alpha,\beta)$.
\end{proposition}

\begin{proof}
See Appendix B.
\end{proof}

\section{Particle mechanics}

In this section, we will demonstrate that the QW model, the QW as a quantum mechanical toy model, naturally contains the correspondence between classical mechanics and {\lq\lq}quantum-like behavior{\rq\rq}.

\begin{figure}
\begin{center}

\includegraphics[width=0.7\hsize]{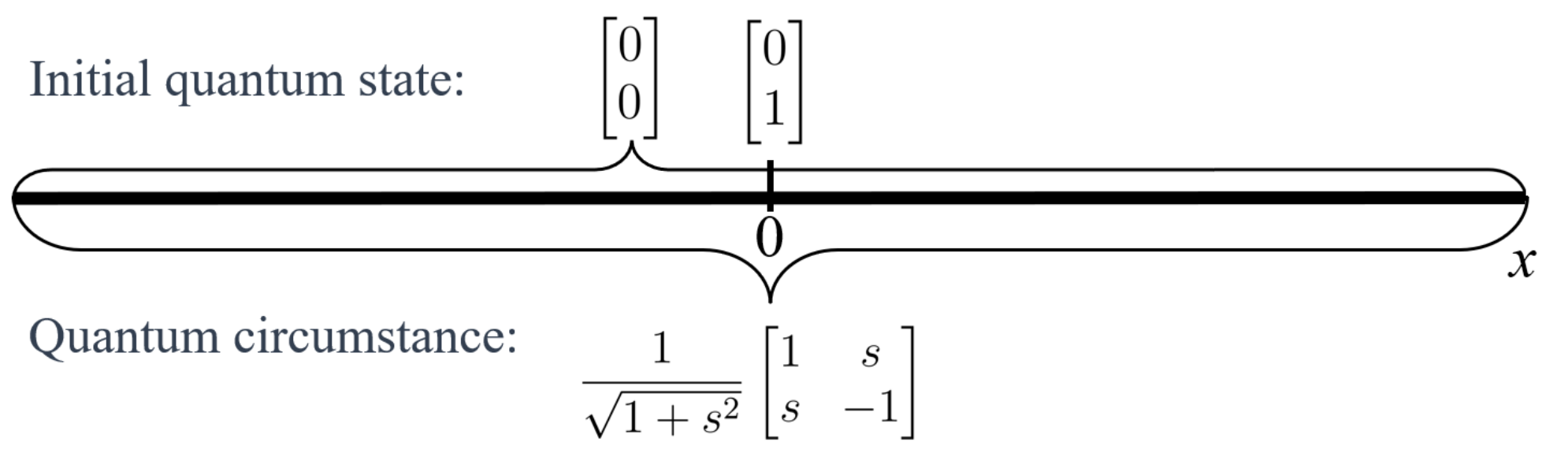}
\

\smallskip

\end{center}

\caption{
Constant quantum circumstance with quasi-inner scattering ratio $s$.
}\label{fg:system_sKdist}
\end{figure}

For this purpose, we consider the constant quantum circumstance $U_x=U$ for all $x\in\ZZ$, with the quantum coin
\[
U=\frac{1}{\sqrt{1+s^{2}}}\left[\begin{matrix} 1 & s \\ s & -1\end{matrix}\right]\in\rU(2),\qquad s\ge 0 .
\]

Depending on $s$, the behavior of wave profile of the QW exhibits the following characteristics:

Since for a certain initial condition, we find that the probability density has the characteristic profiles, as a special case of Konno distribution in Theorem  \ref{th:Konno}, we set the initial condition,
\begin{equation}
\Psi(0,x) =  \varphi_r\delta_{x,0}, \quad
\varphi_r:=\displaystyle{\left[\begin{matrix} 0 \\ 1 \end{matrix}\right]}.
\label{eq:ICparticle}
\end{equation}
Then, we observe the wave profile that has a peak in the right hand side $x>0$.
We refer to the peak as {\lq\lq}wave front{\rq\rq} as displayed in Figure \ref{fg:NEr} obtained by numerical evaluation.
Figure \ref{fg:NEr} (a) shows the probability density function $\rho(t,x)$ for $s=0.1$, and (b) shows the profile for $s=2$.
Both provide a peak, obtained by a special case of the Konno distribution, which is obtained for $s>0$ by setting $\alpha=0$, $\beta=1$ in Theorem \ref{th:Konno},
\begin{equation}
f_K(x) = \frac{s}
  {\pi (1-x^2)\sqrt{1-(1+s^2)x ^2}}
 I_{(-1/\sqrt{1+s^2},\,1/\sqrt{1+s^2})}(x) \times
\left( 1+x \right).
\label{eq:1.1}
\end{equation}
The shape of (\ref{eq:1.1}) is illustrated in Figure \ref{fg:Dep_s} (a).
We refer to the point of divergence on the right-hand side as the wave front corresponding to Figure \ref{fg:NEr}.
The speed of the peak or the wave front in the case is also determined as the result of  Theorem \ref{th:Konno}.

\bigskip

We will interpret these results as follows.

\begin{enumerate}

\item 
The wave front can be defined well by the Konno distribution for the initial state (\ref{eq:ICparticle}), as shown in the numerical computation in Figure \ref{fg:NEr}.
Although the Konno distribution diverges at $x = \pm |a|=\pm 1/\sqrt{1+s^2}$ in Theorem \ref{th:Konno}, its contribution to the integral is finite.
Further due to the property of the Konno distribution, the wave front has the  constant velocity $c=|a|$.

\item For $s=0$, the matrices in (\ref{eq:P,Q}) reduce to
$P_x=\left[\begin{matrix} 1 & 0 \\ 0 & 0\end{matrix}\right]$ and
$Q_x=\left[\begin{matrix} 0 & 0 \\ 0 & -1\end{matrix}\right]$,
so that (\ref{eq:PQPsi}) decouples into
\begin{equation}
\Psi_{L}(t+1,x)=\Psi_{L}(t,x+1), \qquad \Psi_{R}(t+1,x)=-\Psi_{R}(t,x-1). 
\label{eq:PQPsis=00}
\end{equation}

For the initial state (\ref{eq:ICparticle}) we have $\Psi_L = 0$ and
$\Psi_R(t,x)=(-1)^t\delta_{x,t}$: the state is transported to the right
with speed $1$ without spreading.
Writing $\Psi_{R}(t,x)=(-1)^t f(t,x)$, the continuum limit of the second
equation in (\ref{eq:PQPsis=00}) is the advection equation,
$$
  \left[\frac{\partial }{\partial t}+ \frac{\partial}{\partial x}\right] f=0.
$$
In other words, the $s=0$ case describes a trivially moving particle on the lattice $\ZZ^2$; the speed $c$ of the particle is $c=1$ [space step]/[time step].

\item Even for $s\neq 0$ but smaller $s > 0$, it also exhibits a particle-like picture as in Figure \ref{fg:NEr} (a).
Even for $s>0$ case, we also find the particle-like property of the wavefront of the QW:
Its velocity remains constant over time, moving straightforwardly to the right depending upon the initial state:
If we put $\varphi_r:=\displaystyle{\left[\begin{matrix} 1 \\ 0 \end{matrix}\right]}$, we locate the wave-front moving to the left.

For $s \to +0$, it shows the deterministic classical particle with velocity $c=1$[space step]$/$[time step], and the probability density shrinks to the $\delta$-like distribution, which is the same as the random walk limit of the deterministic limit.

However it is very surprising that the development of the function (physically, scalar field or the wave function) represents the trajectory of the particle well even for $s>0$.

\item For $s>0$, the speed of the wave front $1/\sqrt{1+s^2}$ is less than $1$ as in Figure \ref{fg:Dep_s} (b);
This brings to mind the relationship between group velocity and phase velocity, which is associated with the de Broglie-Bohm pilot wave theory \cite{Bo,Bu}, although, the QW model does not provide a deterministic picture or hidden variables.

\item For larger $s$, it behaves like wavy profile.
When $s$ increases, the peak widens and the wavefront propagates more slowly, as illustrated in Figure \ref{fg:Dep_s} (b).

\item At a position $x \in \ZZ$ around the wave front such that $\bigl||a|t - x\bigr| \in O(t^{1/3})$, the finding probability is estimated by $\rho(t,x) \in O(t^{-2/3})$ \cite{ST, MSSS}, while at a position $x \in \ZZ$ such that $|x| < |a|t$, it is estimated by $\rho(t,x) \in O(t^{-1})$, where $a=1/\sqrt{1+s^{2}}$ is the $(1,1)$ entry of $U$ in the notation of Theorem \ref{th:Konno}.
Here, for functions $h$ and $g$ of $t$ such that $g(t)>0$, we write $h\in O(g(t))$ $(t\to\infty)$ if there are constants $C,T>0$, independent of $t$, such that $|h(t)|\le C\,g(t)$ for $t\ge T$, i.e., $h(t)/g(t)$ is bounded as $t\to\infty$.

\end{enumerate}

In other words, we find that the QW model has the parameter $s$ that bridges the states between the particle-like and classical particle limits and the wavy picture.

We will proceed to investigate the properties of the QW model as follows.
In Section \ref{sec:3.1}, we demonstrate that the model exhibits particle-like properties even in the case of reflection and partial reflection. In Section \ref{sec:4}, we show that the model exhibits the interference phenomenon.

\bigskip

\bigskip
\begin{figure}
\begin{center}

\includegraphics[width=0.4\hsize]{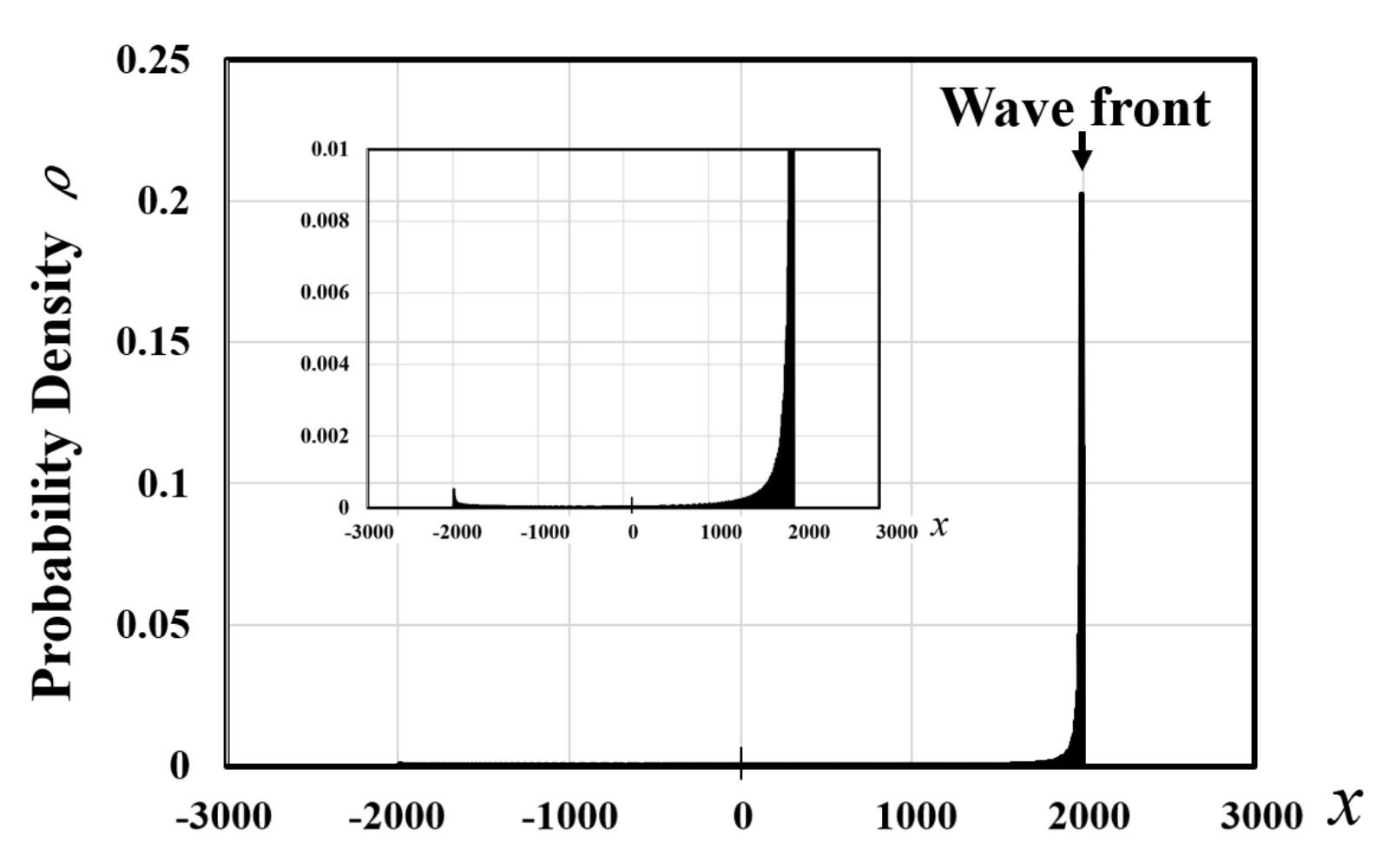}
\hskip 0.05\hsize
\includegraphics[width=0.4\hsize]{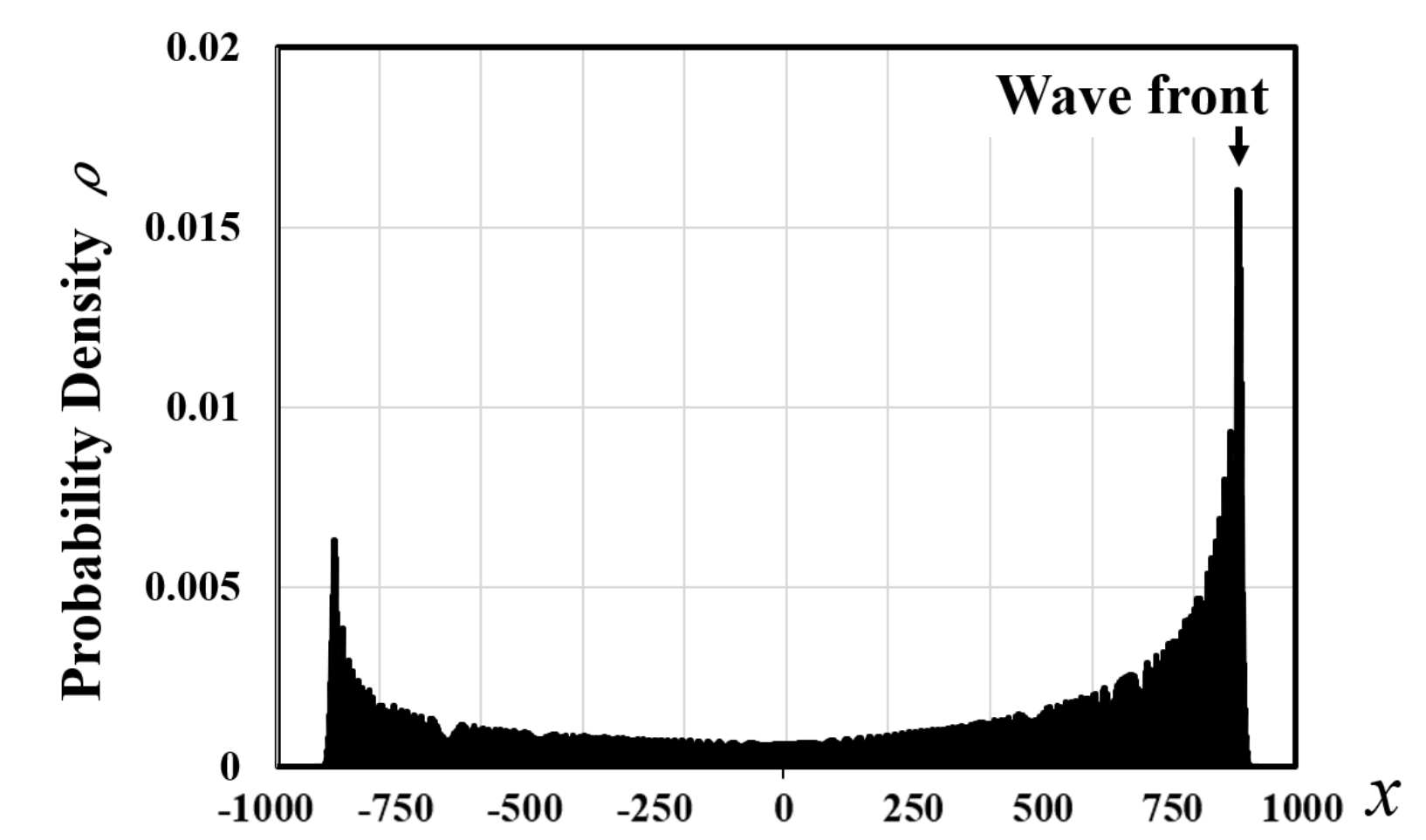}

(a)
\hskip 0.45\hsize
(b)
\hskip 0.1\hsize\

\smallskip
\smallskip
\smallskip

\end{center}

\caption{
Dependence on $s$:
the numerical calculation shows (a) for $s=0.1$ and (b) for $s=2$.
}\label{fg:NEr}
\end{figure}

\bigskip
\begin{figure}
\begin{center}

\includegraphics[width=0.4\hsize]{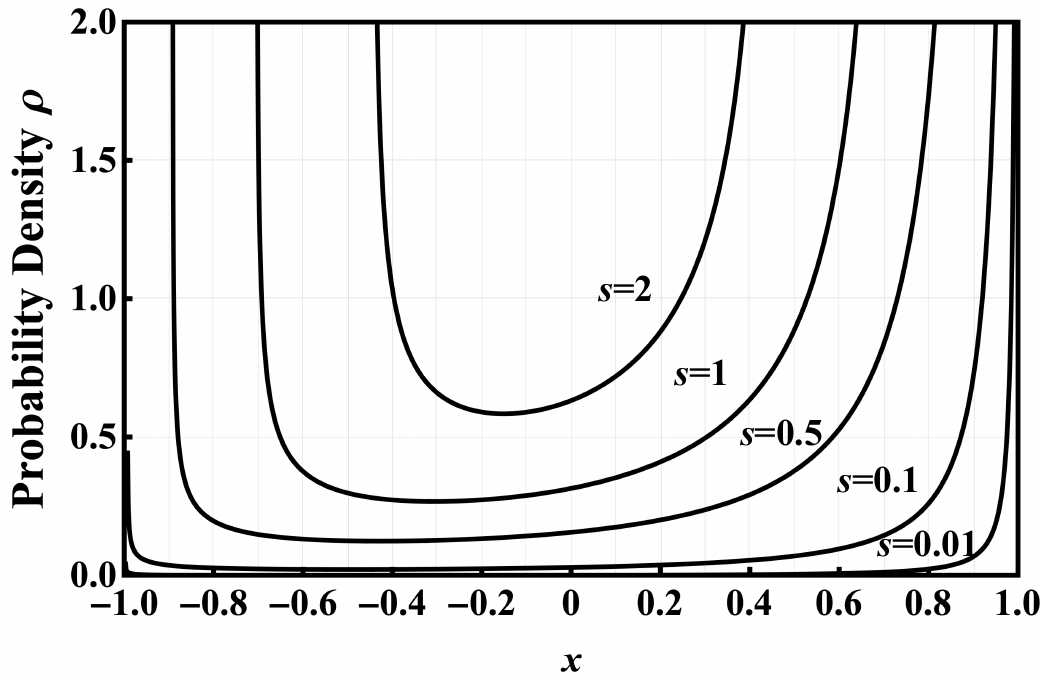}
\hskip 0.05\hsize
\includegraphics[width=0.4\hsize]{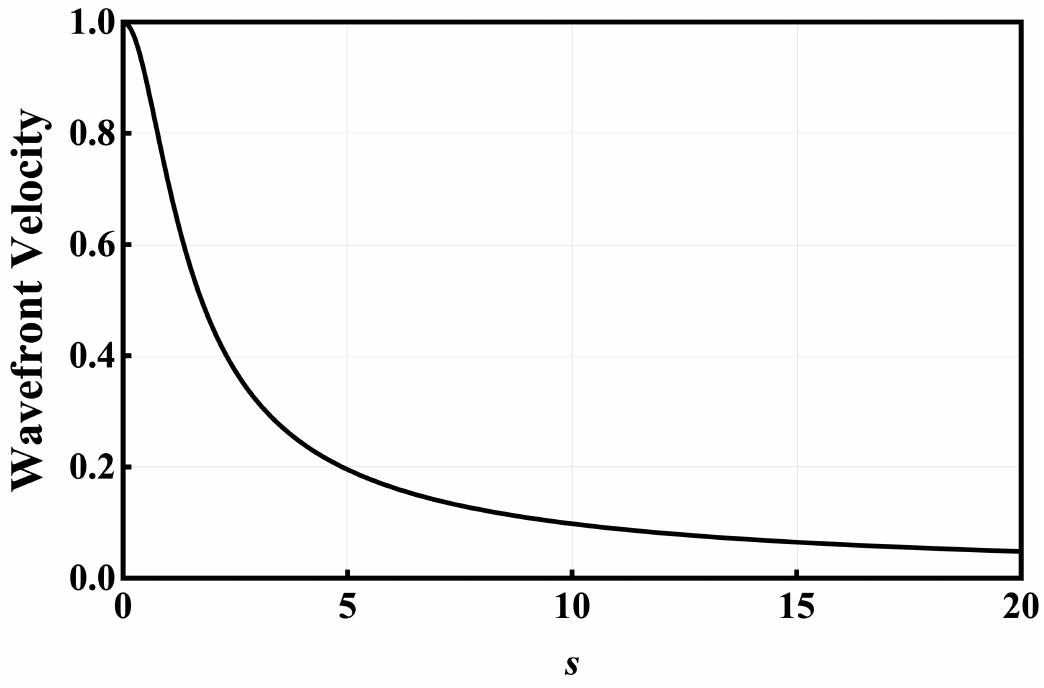}

(a)
\hskip 0.45\hsize
(b)
\hskip 0.1\hsize\

\smallskip
\smallskip
\smallskip

\end{center}

\caption{
Dependence on $s$ in (\ref{eq:1.1}) (a) and 
the dependence of the velocity of the wavefront on $s$ (b)
}\label{fg:Dep_s}
\end{figure}

Before presenting them, we will provide some comments on the QW model and the optics, as mentioned in part in \cite{IKMM}.

\begin{remark}\label{rmk:Zigzag}
{\rm{
As we give a short review of Feynman's checkerboard Dirac model in Appendix A, 
we give some comments on the relation between Feynman's model and the QW model.

For the checkerboard (zigzag) diagram, if we choose two endpoints, P and Q, then the length of each path is equal.
Feynman used this fact and added the phase $i \varepsilon$ at each turn in his model;
 the development is not obtained by the unitary matrix contrast to the QW model.

Even in Feynman's model of the Dirac particle, we can approximate the Euclidean distance in the zig-zag system by assigning the phase $i\varepsilon$ at every turn, from right to left or from left to right, as in Appendix A.

We note that in his model, unlike our QW model, unitarity is not guaranteed in the time evolution.
Therefore, his model does not yield desirable properties, such as the nature of a probability distribution as a time-dependent process.
It is obvious that, unfortunately, results like the QW cannot be obtained using his model.

However, when we discretize the wave system, we can recall the theories of scattering (S-)matrix and transfer (T-)matrix in optics \cite{Crook} and the one-dimensional quantum mechanics \cite{Schiff}.
There, we discriminate the right-direction state and the left-direction state at each point.
Furthermore, based on the theories of S and T matrices in optics, we assign a two-dimensional vector space to each point.
Crook proved that the sum of the optical lengths of the paths with reflection and transmission rate in these system recovers the S-matrices derived from the one-dimensional Helmholtz (or Schr\"odinger) equation, as the path integral method recovers the Schr\"odinger equation \cite{Schul}.
Crook's picture of the path sum method is essentially the same as the primitive model of Feynman's path integral method \cite{Crook}.

Furthermore, in the S-matrix theory, we observe a correspondence between the incoming and outgoing rays.
We regard this as causality because the outgoing rays are caused by the incoming rays.
Thus, we can interpret this as a time-based process whereby incoming rays generate outgoing rays at a later moment.
Then, we can recreate the QW model by assuming that the particle is scattered at every point and that the S-matrix is unitary as partially mentioned in \cite{IKMM}.

Moreover we give a comment on the particle-like nature and the path integral method.
We demonstrate Fermat's principle in optics using the QW model, in which the wavefront propagates linearly in space-time.
This picture was predicted by Feynman \cite{F}.
Feynman stated that, by considering the sum of all paths of a photon from points $P$ and $Q$ with weight $e^{2\pi\ii L_c/\lambda}$, Fermat's principle can be recovered as a thought experiment, where $L_c$ is the optical length along the path and $\lambda$ is the wavelength of the photon.

In the checkerboard-diagram, the path-length from time $t$ to $t + \Delta t$ is determined as a value proportional to $\Delta t$ as in Appendix A.
Thus the optical length cannot be expressed well in the checkerboard-diagram.
However, the number of combinations represents the distance between $(x_0, t)$ and $(x_1, t + \Delta t)$.
In fact, Konno discovered the Konno distribution by counting the number of combinatorial paths and taking into account the non-commutativity of $P_x$ and $Q_x$.
It can be considered the mathematical realization of Feynman's thought experiment for the Fermat principle in \cite{F}.

We note that Gudder also studied a similar model of QW for a similar motivation to understand quantum mechanics \cite{Gud}, and Strauch studied a relation between QW models inspired by Feynman's checkerboard and the Dirac equation \cite{Str}.

}}
\end{remark}

\begin{remark}\label{rmk:SunadaTate}
{\rm{
The asymptotic behavior of the probability density function of QW, in accordance with Theorem \ref{th:Konno} weakly, describes particle-like behavior as mentioned above. However, if we zoom in on the distribution of the function around the peak of the probability density function, as shown in Figure \ref{fg:ST_dist}, we can see that it also exhibits wave-like properties.

Sunada and Tate also showed more precise profile of the QW as in \cite{ST}.
They described the profile of the outer side of the peak illustrated in Figure \ref{fg:ST_dist} by using the Airy functions after investigation of the perturbation of the Konno distribution Theorem \ref{th:Konno}.

This reminds us of the refinement of the caustic phenomenon when considering the wave effect, as discussed in Section \ref{sec:5} (see Figure \ref{fg:BU}).
}}
\end{remark}

\bigskip
\begin{figure}
\begin{center}

\includegraphics[width=0.6\hsize]{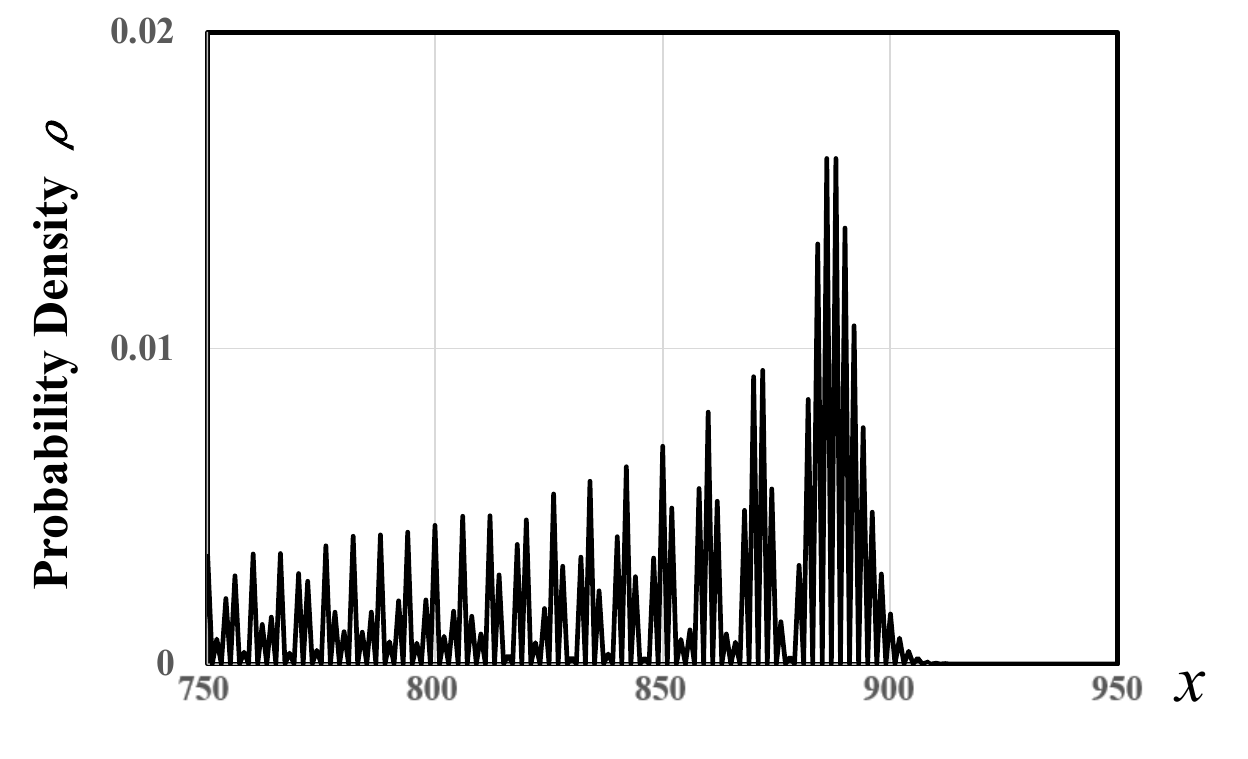}
\

\end{center}

\caption{
The distribution $\rho(t,x)$ of the Hadamard walk around its peak.
}\label{fg:ST_dist}
\end{figure}

\bigskip

\subsection{Numerical observation of its partial-reflection phenomena}
\label{sec:3.1}

\begin{figure}
\begin{center}

\includegraphics[width=0.7\hsize]{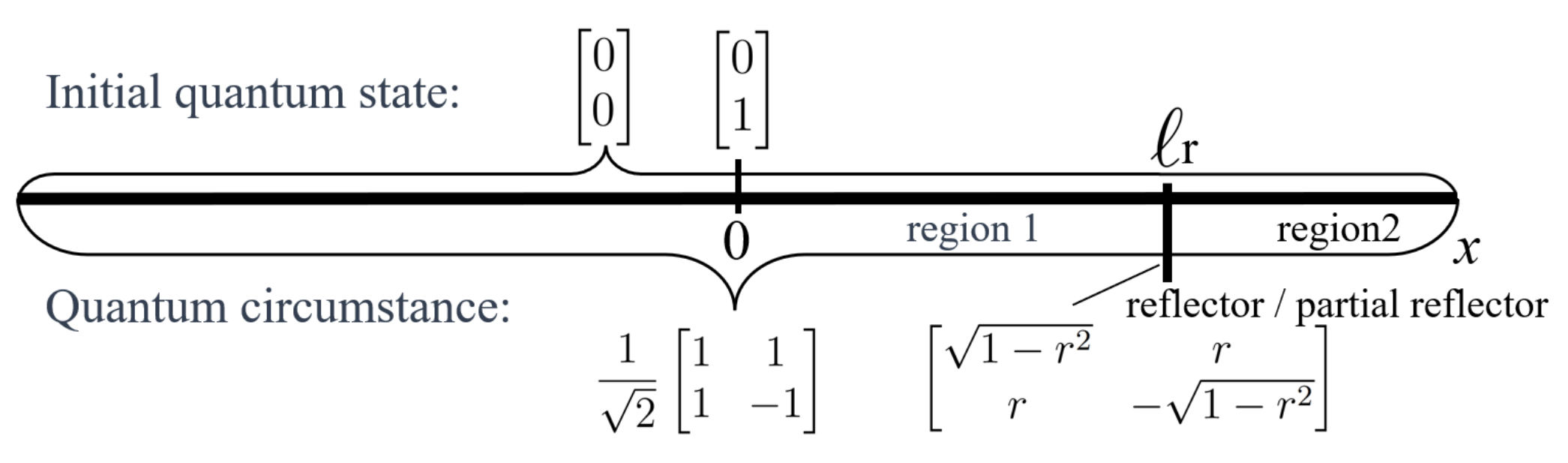}
\

\smallskip

\end{center}

\caption{
Constant quantum circumstance with a half-reflection plate
}\label{fg:system_NSimDist}
\end{figure}

In this subsection, we demonstrate that the QW model also represents the reflection and partial-reflection of particle.
Figure \ref{fg:reflect} shows that the wavefront is reflected as a ballistic particle and partially reflected as a wavy wavefront.

In this subsection, we consider the reflection and partial-reflection problem of the QW model.
As shown in Figure \ref{fg:system_NSimDist}, we set the quantum circumstance $\cU$ with $U_x$
\begin{equation}
     U_x = \left\{\begin{matrix} 
    \left[\begin{matrix} \sqrt{1-r^2} & r \\ r & -\sqrt{1-r^2}\end{matrix}\right] & x=\ell_\mathrm{r}>0,\\
\frac{1}{\sqrt2}
    \left[\begin{matrix} 1 & 1 \\ 1 & -1\end{matrix}\right] & \mbox{otherwise,}
\end{matrix}\right.
\end{equation}
where $r$ is the reflection rate.
The $r=1$ case means the reflection at the point $\ell_\mathrm{r}$ and $0<r<1$ case shows the partial reflection.

\bigskip
\begin{figure}
\begin{center}

\includegraphics[width=0.37\hsize]{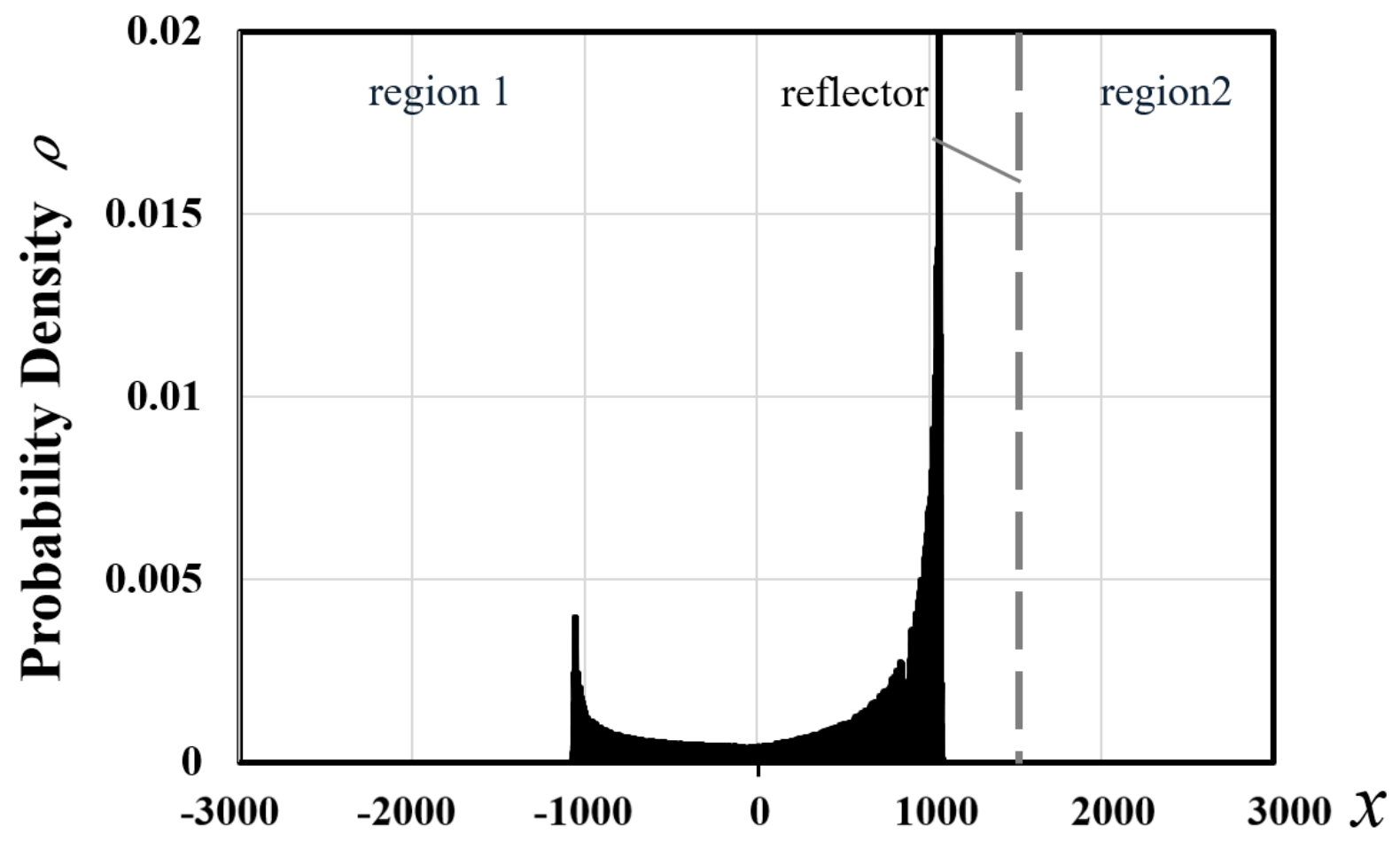}
\hskip 0.01\hsize
\includegraphics[width=0.37\hsize]{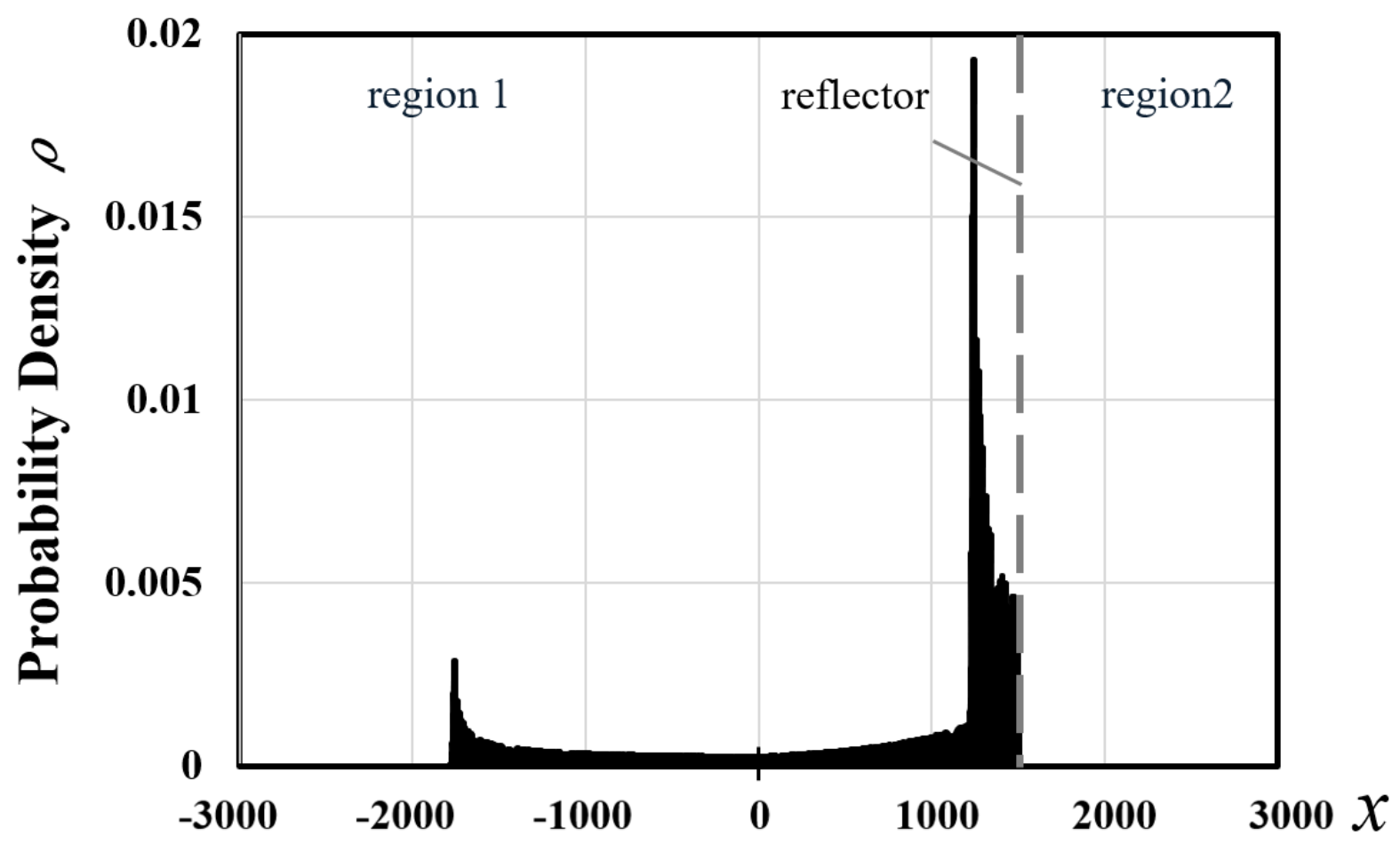}

(a)
\hskip 0.41\hsize
(b)

\includegraphics[width=0.37\hsize]{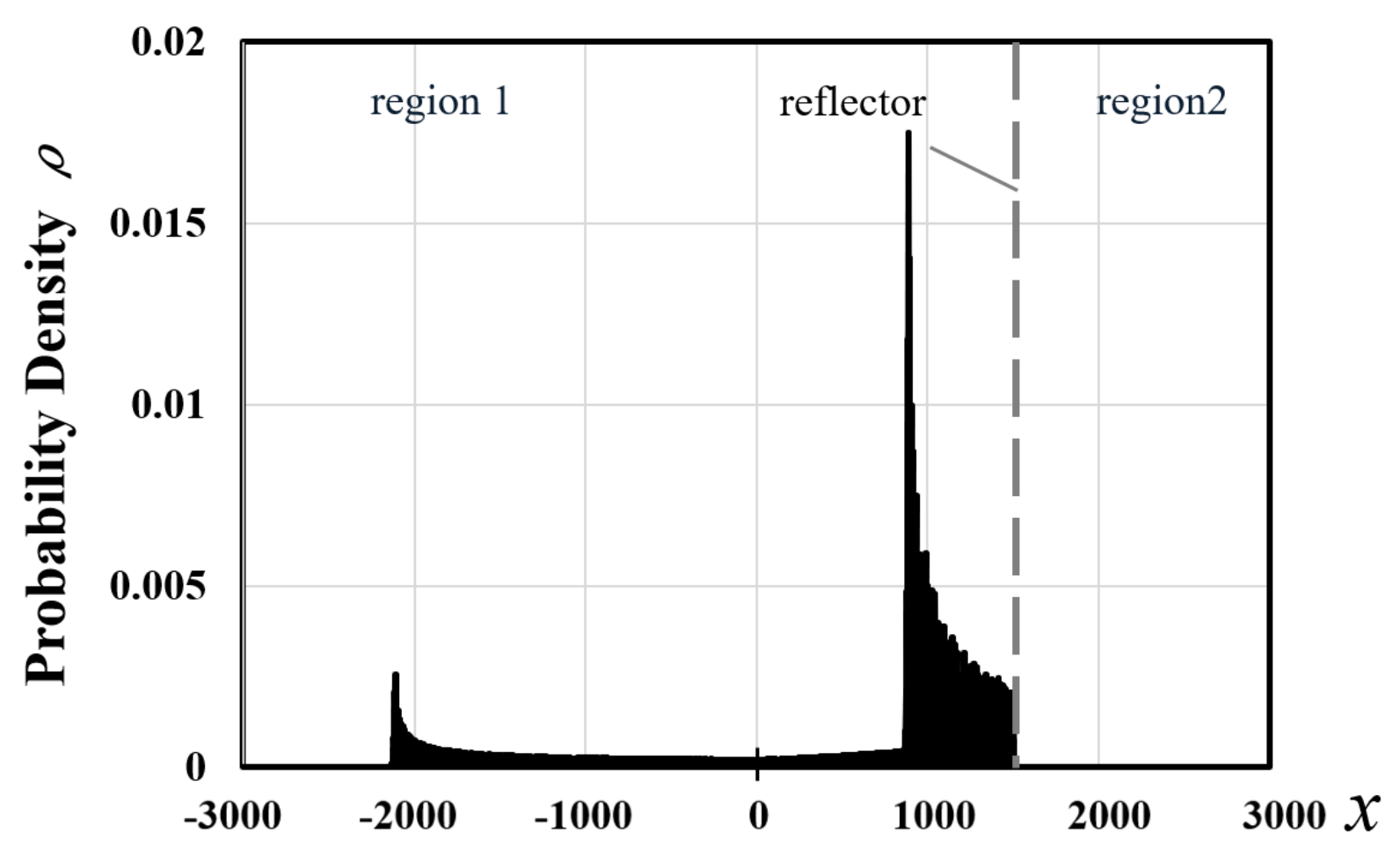}
\hskip 0.01\hsize
\includegraphics[width=0.37\hsize]{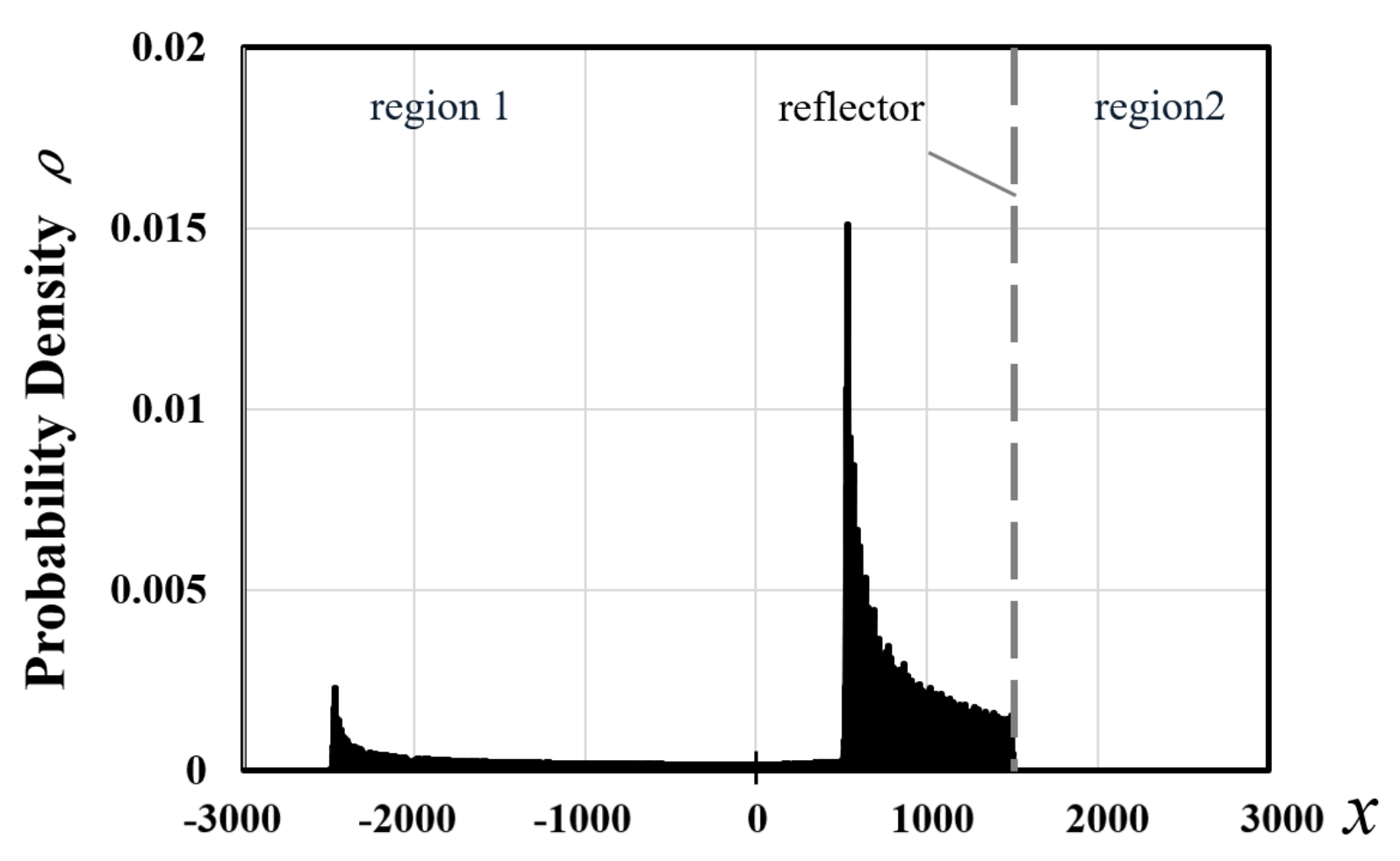}

(c)
\hskip 0.41\hsize
(d)

\smallskip
\smallskip

\includegraphics[width=0.37\hsize]{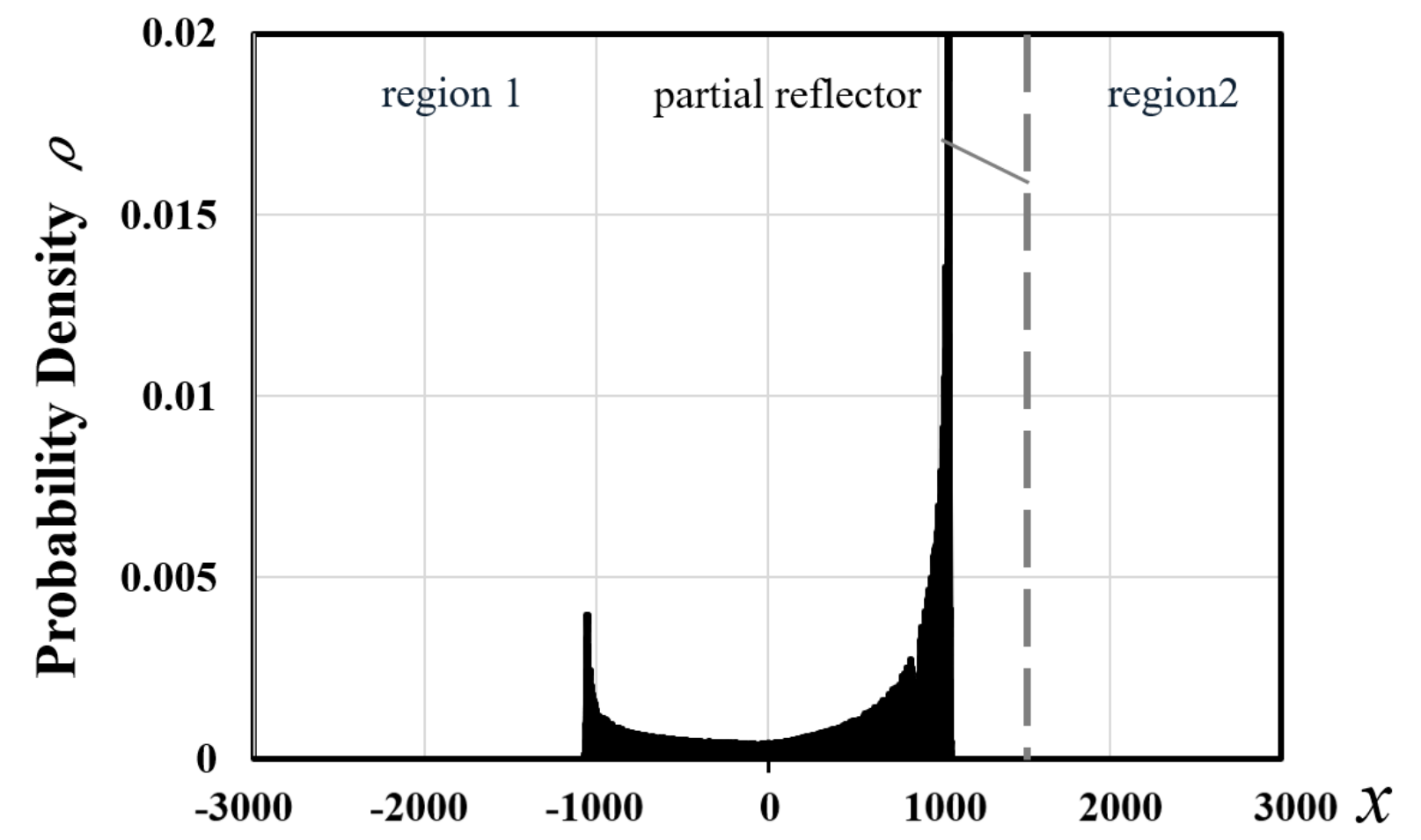}
\hskip 0.01\hsize
\includegraphics[width=0.37\hsize]{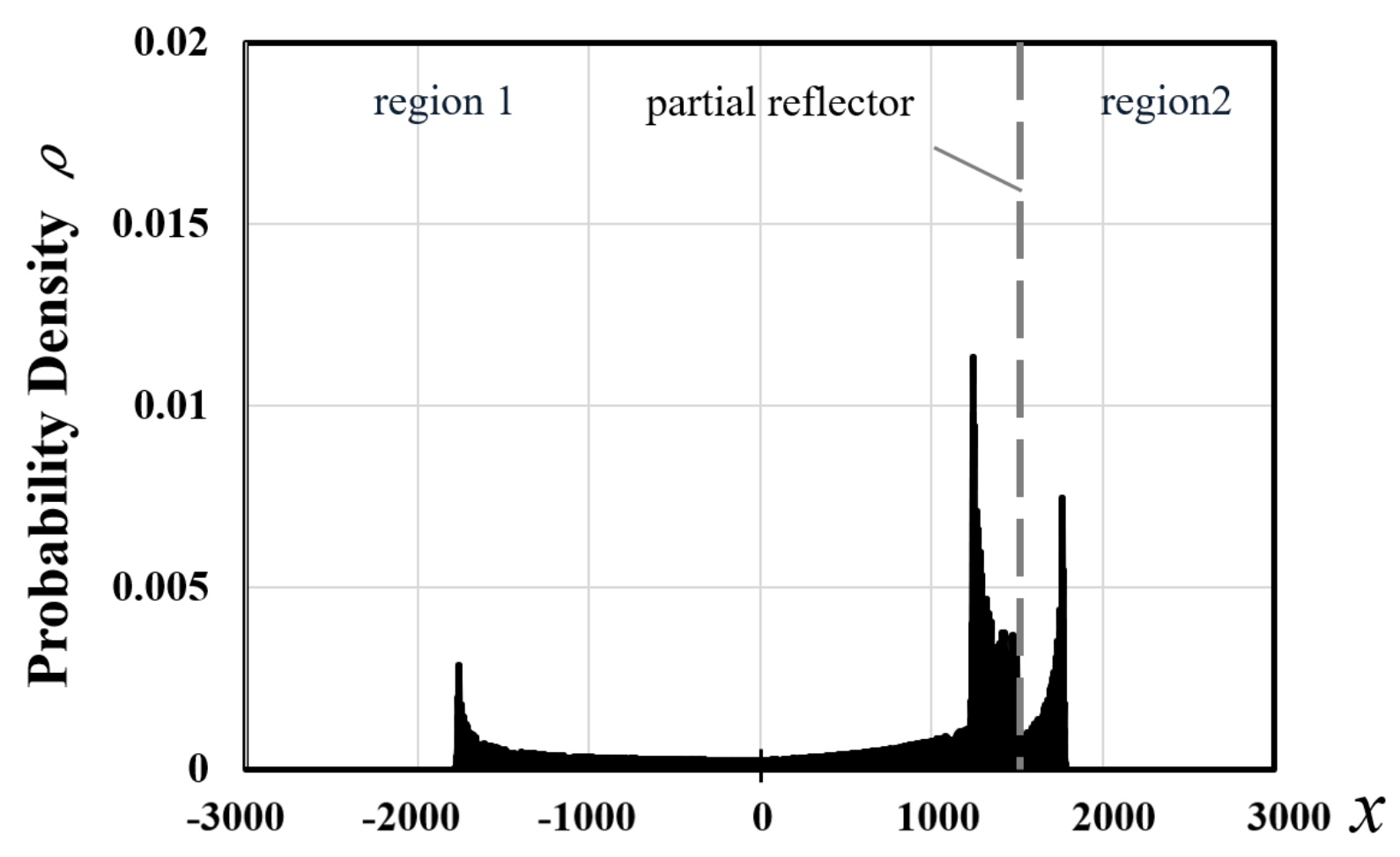}

(e)
\hskip 0.41\hsize
(f)

\includegraphics[width=0.37\hsize]{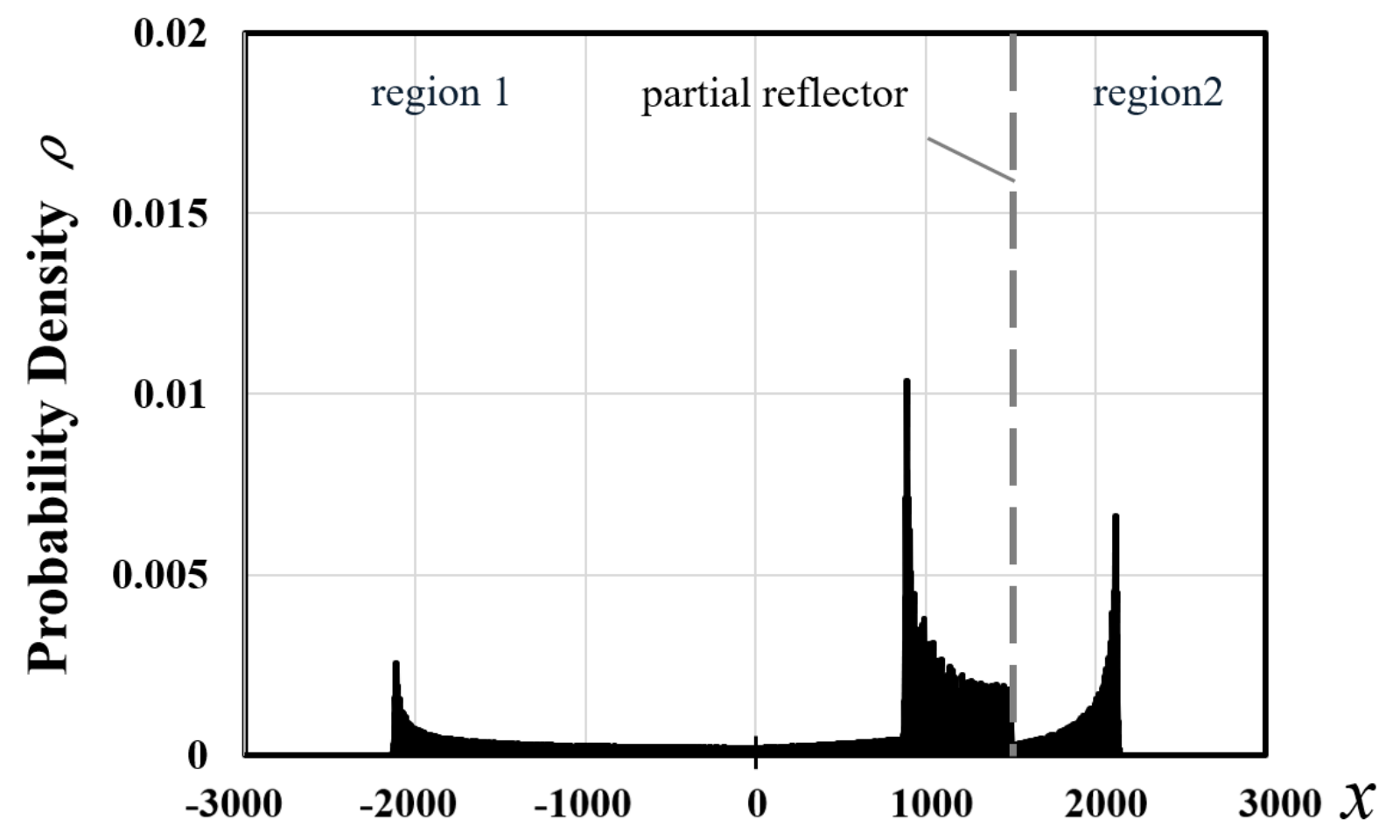}
\hskip 0.01\hsize
\includegraphics[width=0.37\hsize]{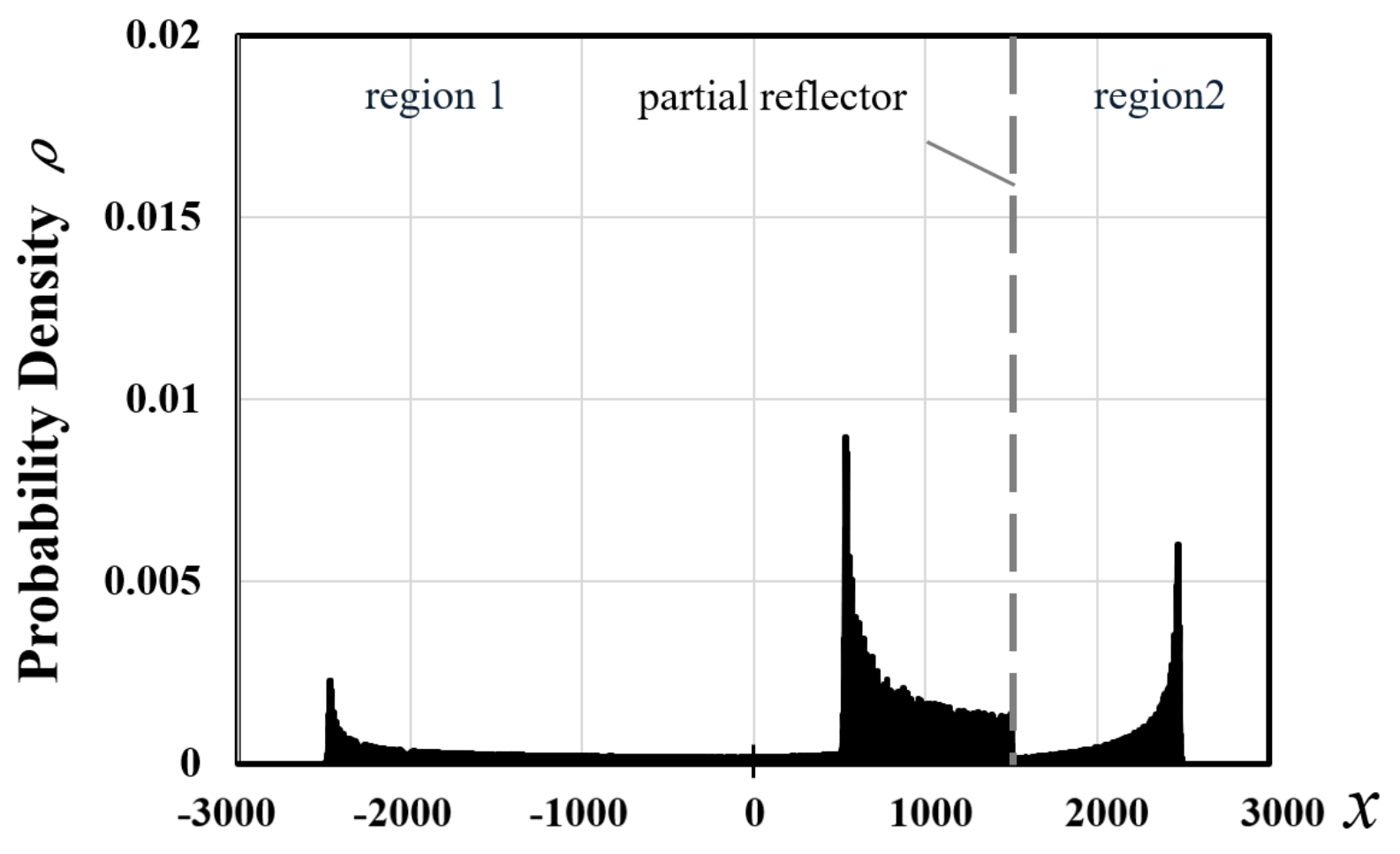}

(g)
\hskip 0.41\hsize
(h)

\smallskip

\end{center}

\caption{
The time evolution of the reflector and partial-reflector of the QW:
(a):$t=1500$, (b):$t=2500$, (c):$t=3000$, (d):$t=3500$ are for the total reflection case whereas
(e):$t=1500$, (f):$t=2500$, (g):$t=3000$, (h):$t=3500$ are for the partial reflection case $r=0.95$.
}\label{fg:reflect}
\end{figure}

As in Figure \ref{fg:reflect}, the QW model demonstrates the ballistic and particle-like behavior of the wavefront at $t=1500$, $t=2500$, $t=3000$  and $t=3500$. 
In the case of the reflector, the wavefront reflects completely. 
Figure \ref{fg:Trj} shows the trajectories of the maximum points in the regions, or the behavior of the wavefront position.
In the case of the partial-reflector $r=0.95$,  the wavefront reflects partially, while the other part passes through transparently.
The behavior in Region 2 is very similar to half of the original Konno distribution.
In fact, even in Region 2, the wavefront position moves like a classical particle as in Figure \ref{fg:Trj} as the dotted guide line shows even in  Figure \ref{fg:Trj} (b).
These are straight lines; the orbits show that the wavefront moves along the path that minimizes the distance between two points in spacetime.
In other words, it obeys Fermat's principle, as explained by Feynman in his book \cite{F}.
Furthermore, at the point of partial reflection, the wavefront of the transmitted wave is smoothly connected to the incoming one.
In other words, it behaves entirely like a particle.
It implies that there is no phase-shift effect after collision, as is encountered in nonlinear wave dynamics.

Moreover, in addition to the fact that these are obtained as the time evolution of a function (scalar field), note that even when the quantum circumstance --such as a reflector or a partial reflector-- undergoes significant non-perturbative changes, the waveform remains undisturbed and, piecewise, can be seen as having the shape of the Konno distribution or a superposition thereof.
In other words, we note that the Konno distribution (\ref{eq:1.1}) provides an extremely robust asymptotic state, which characterizes the quantum walk.

Such properties represent the wave dynamical properties, which are required even though the QW is a mathematical toy model of the path integral.

\bigskip
\begin{figure}
\begin{center}

\includegraphics[width=0.45\hsize]{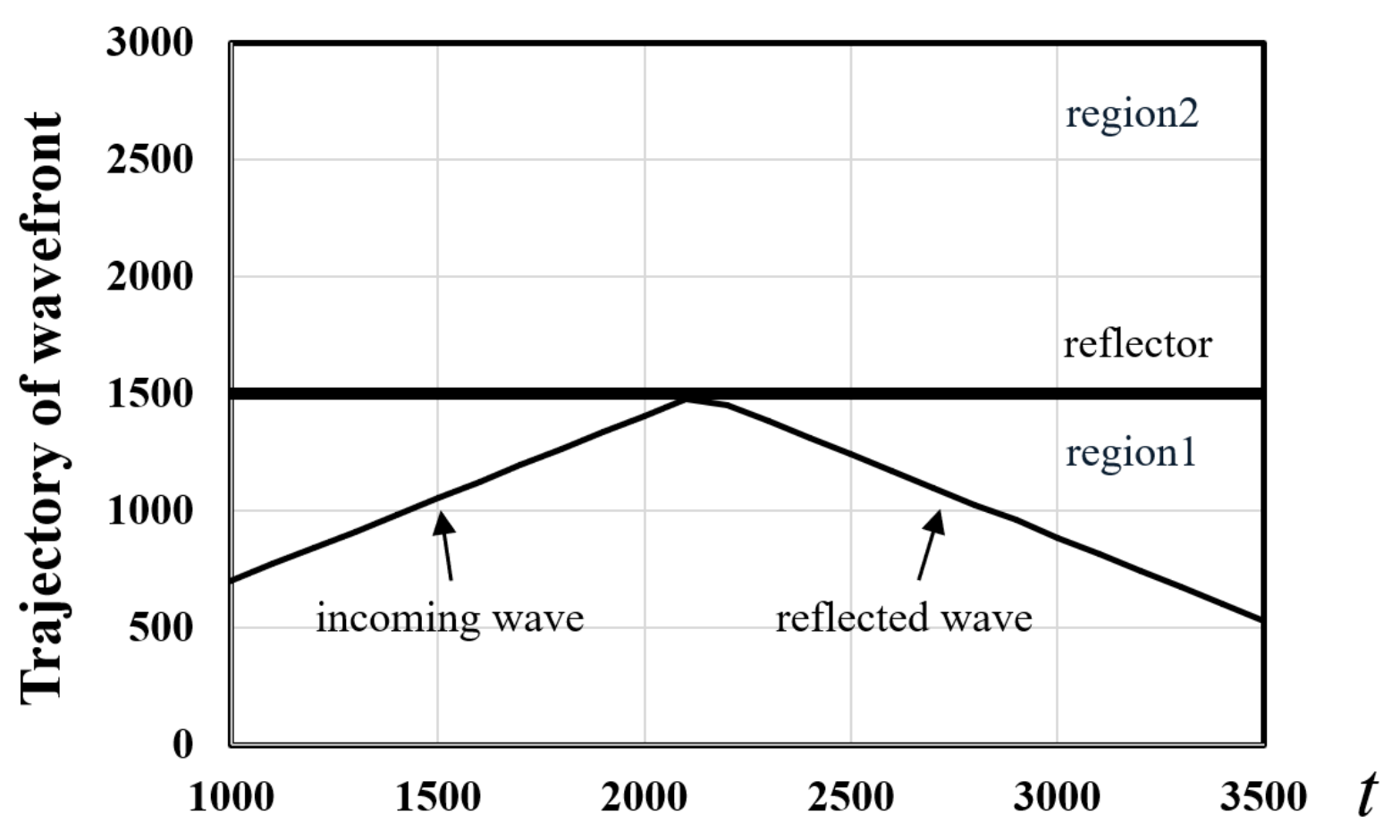}
\hskip 0.01\hsize
\includegraphics[width=0.45\hsize]{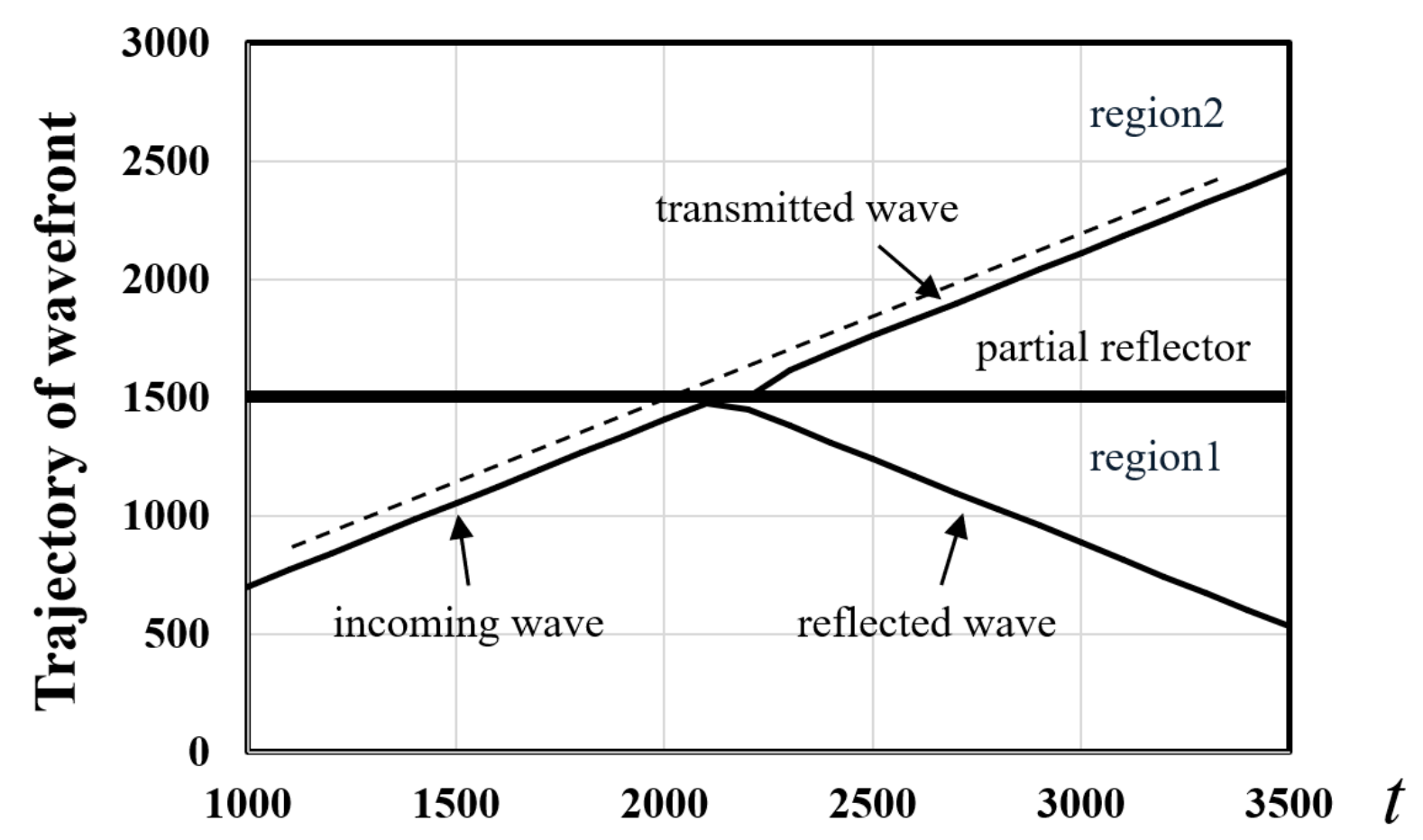}

(a)
\hskip 0.45\hsize
(b)

\smallskip

\end{center}

\caption{
The trajectory of the maximal points in the regions:
They show the behavior of the wavefront for the reflection case and partial reflection case $r=0.95$.
The dotted line is a guide line of trajectory.
}\label{fg:Trj}
\end{figure}

\section{Interference phenomenon}\label{sec:4}

As in the previous section, we demonstrated the particle nature of the QW.
In this section, we demonstrate that the QW model exhibits interference, or the wavy property.
To the best of our knowledge, no report on the interference phenomenon in the QW model has been published.
In other words, by the content of this section, we will show that the QW model exhibits both particle-like (or ballistic) properties, based on the Konno distribution, and wavy properties.
The interference phenomenon we present below has some very interesting properties that are similar to the far-field type diffraction known as the Fraunhofer type \cite{BW, HNSxx} including an X-ray diffraction pattern of a crystal \cite{BW}.
However, it also may be related to a quantization of a dipole via the Bohr quantization condition, and also the near-field type interference related to the number theory, such as the rational Talbot effect \cite{BW, MO} since it is characterized by a natural number $\ell \in \ZZ_{>0}:=\{\ell \in \ZZ\ |\ \ell >0\}$.

\begin{figure}
\begin{center}

\includegraphics[width=0.7\hsize]{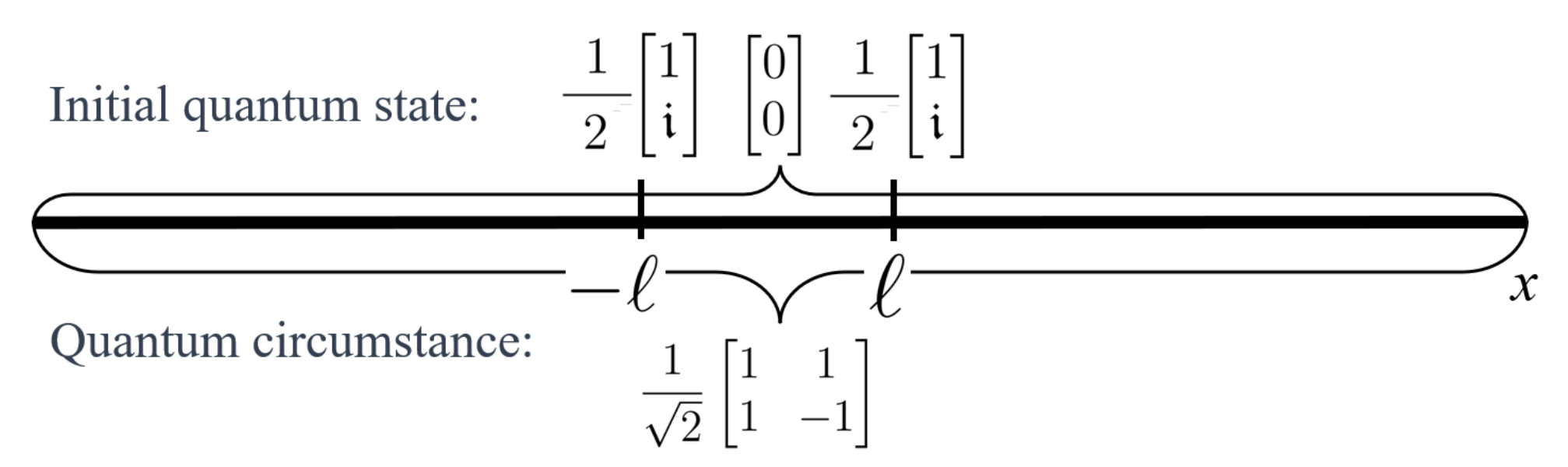}
\

\smallskip

\end{center}

\caption{
Double slit model 
}\label{fg:system_Wslit}
\end{figure}

Let us consider the double slit.
In other words, we consider the initial condition,
$$
\Psi(0,x)
=\frac{\varphi}{\sqrt{2}}\delta_{x,-\ell }+\frac{\varphi}{\sqrt{2}}\delta_{x,\ell},
$$
especially $\varphi=\varphi_c:=\displaystyle{\frac{1}{\sqrt{2}}\left[\begin{matrix} 1 \\ \frki \end{matrix}\right]}$ as in Figure \ref{fg:system_Wslit}.

Figure \ref{fg:Wslit} demonstrates the time evolution of the density distribution $\rho(t,x)$ for the $\ell=4$ case.
Then, we observe the interference as a double slit interference as in Figure \ref{fg:WslitAS}.
Figure \ref{fg:WslitAS} shows the profile of the density distribution $\rho(t,x)$ depending on the distance $2\ell$ between the initial source points. 
As in Figure \ref{fg:Wslit}, we observe that the profile of the density distribution $\rho(t,x)$ immediately transitioned to that shown in Figure \ref{fg:WslitAS}.
In other words, this phenomenon differs from the diffraction phenomenon of double-slit interference in quantum mechanics textbooks, e.g., \cite{FHS}.
The latter is essentially explained by the Fresnel type interference \cite{BW}, in which the profile changes depending on the distance from the source points.

The phenomena we observed are similar to the far-field one, known as the Fraunhofer diffraction.
In other words, the interference in Figure \ref{fg:Wslit} is immediately transmitted to the Fraunhofer-type diffraction as in Figure \ref{fg:WslitAS}.
Even though in the QW-model the wave number is explicitly defined, the distance $2 \ell$ is much smaller than the observed time $t$; both are just positive integers and we can compare them as $t \gg 2\ell$.

\bigskip

Thus, it is reasonable to evaluate the asymptotic behavior of the profile of the density distribution $\rho(t,x)$.
We give one of our main theorems in this paper, which is inspired by \cite{Kon14}.
\bigskip
\begin{figure}
\begin{center}

\includegraphics[width=0.45\hsize]{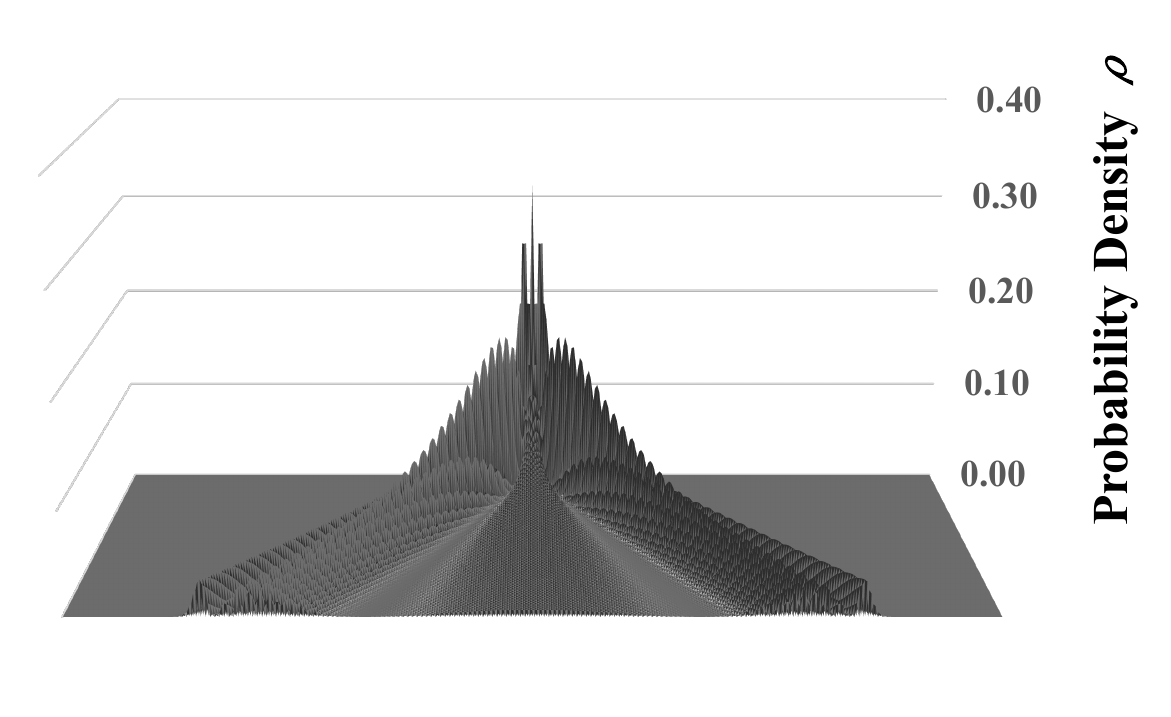}
\hskip 0.05\hsize
\includegraphics[width=0.45\hsize]{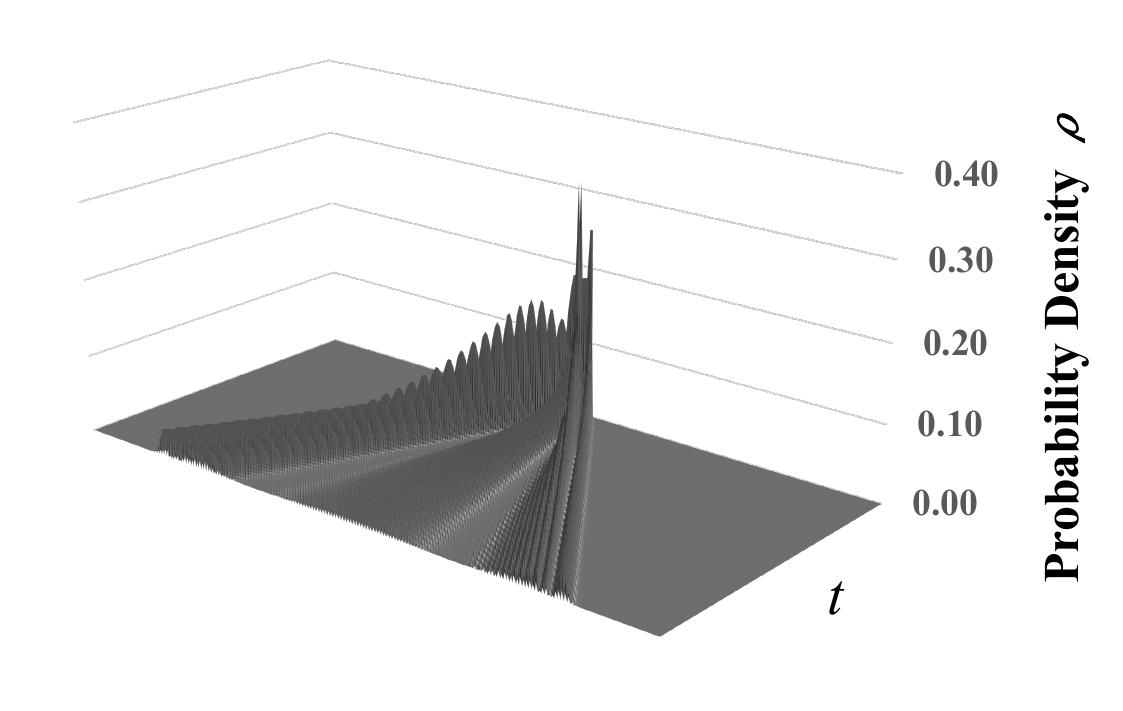}

(a)
\hskip 0.45\hsize
(b)

\smallskip

\end{center}

\caption{
Time evolution of the double slit model in the QW. $\ell=4$
}\label{fg:Wslit}
\end{figure}

\bigskip
\begin{figure}
\begin{center}

\includegraphics[width=0.29\hsize]{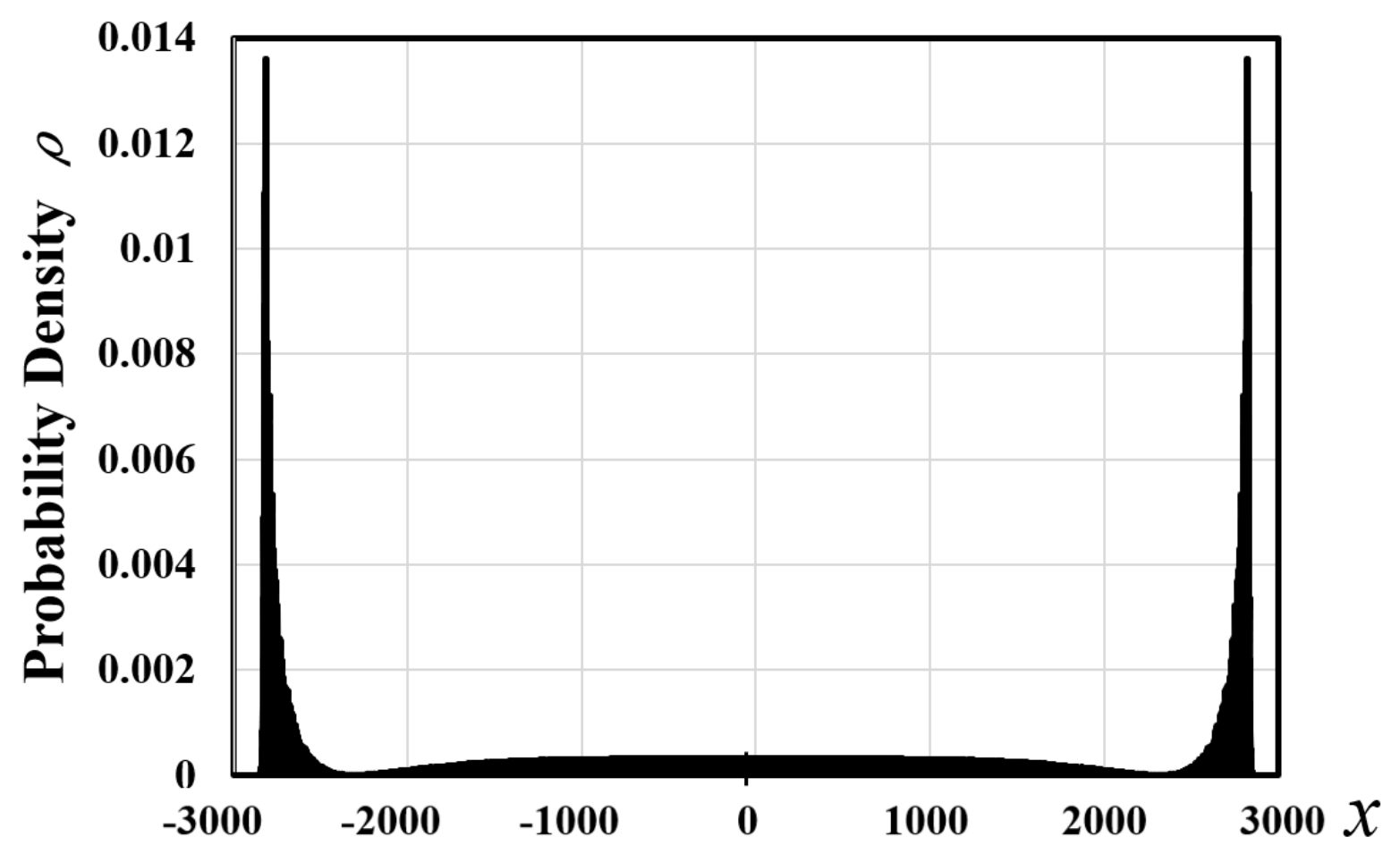}
\hskip 0.05\hsize
\includegraphics[width=0.29\hsize]{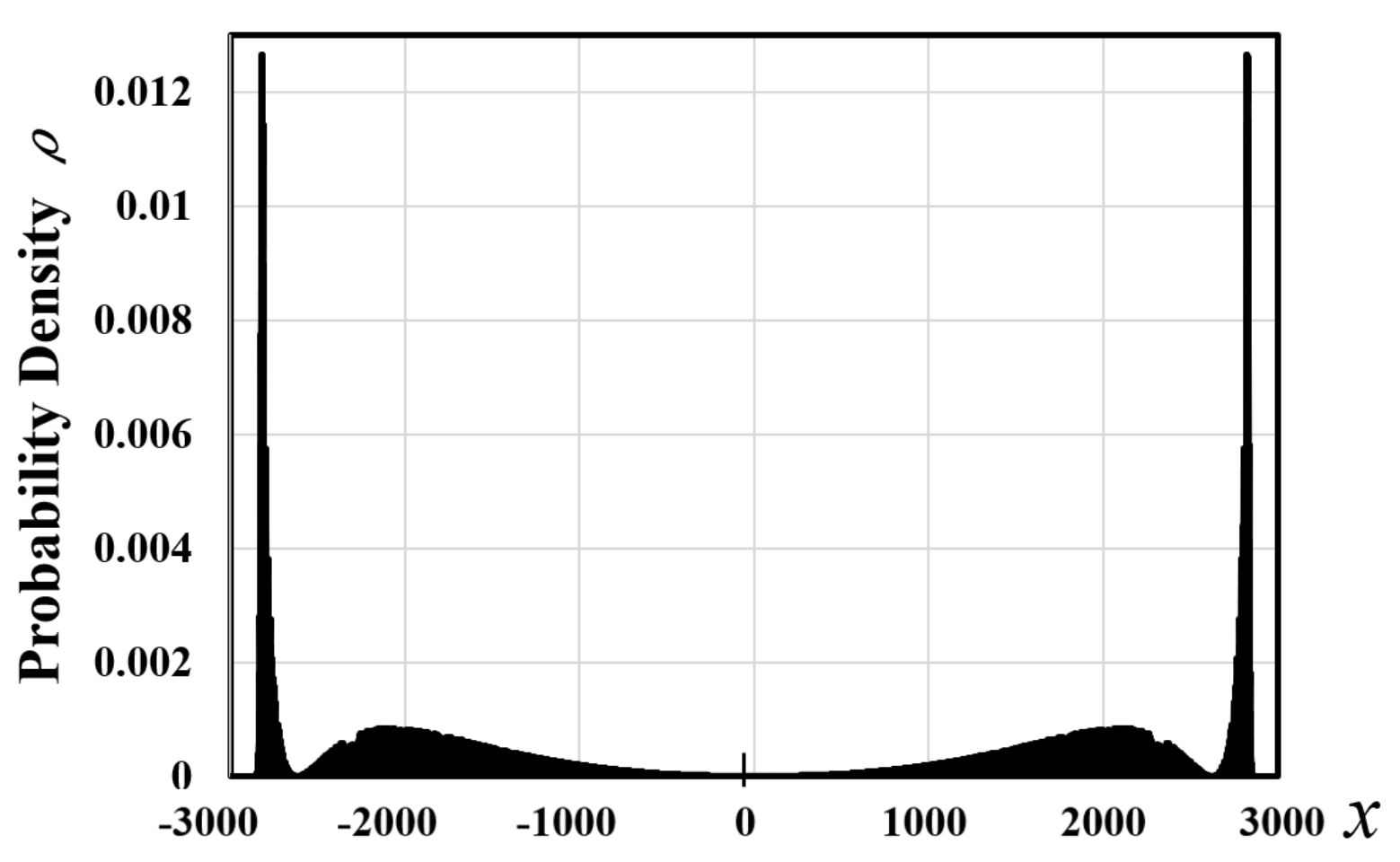}
\hskip 0.05\hsize
\includegraphics[width=0.29\hsize]{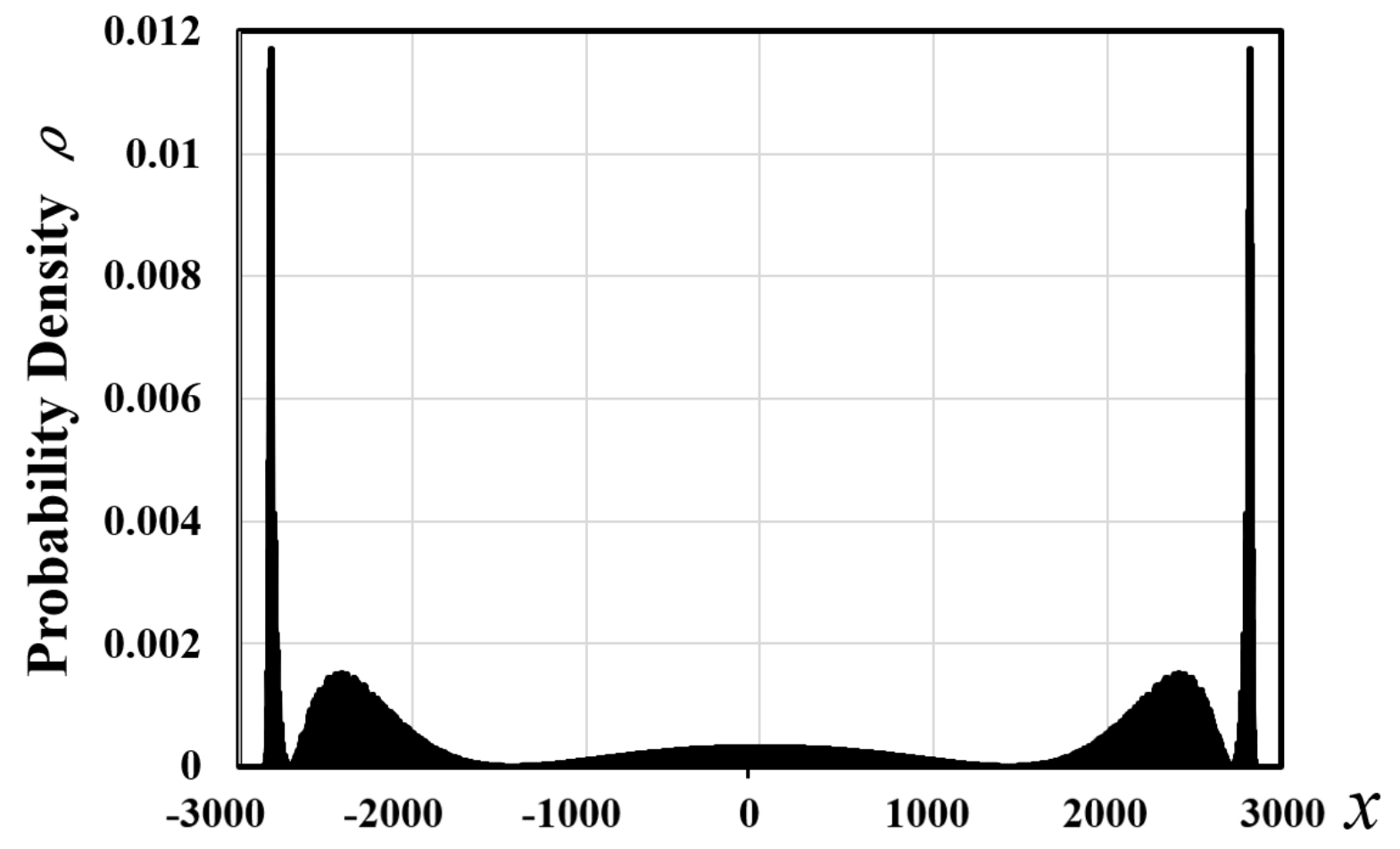}

(a)
\hskip 0.29\hsize
(b)
\hskip 0.29\hsize
(c)

\includegraphics[width=0.29\hsize]{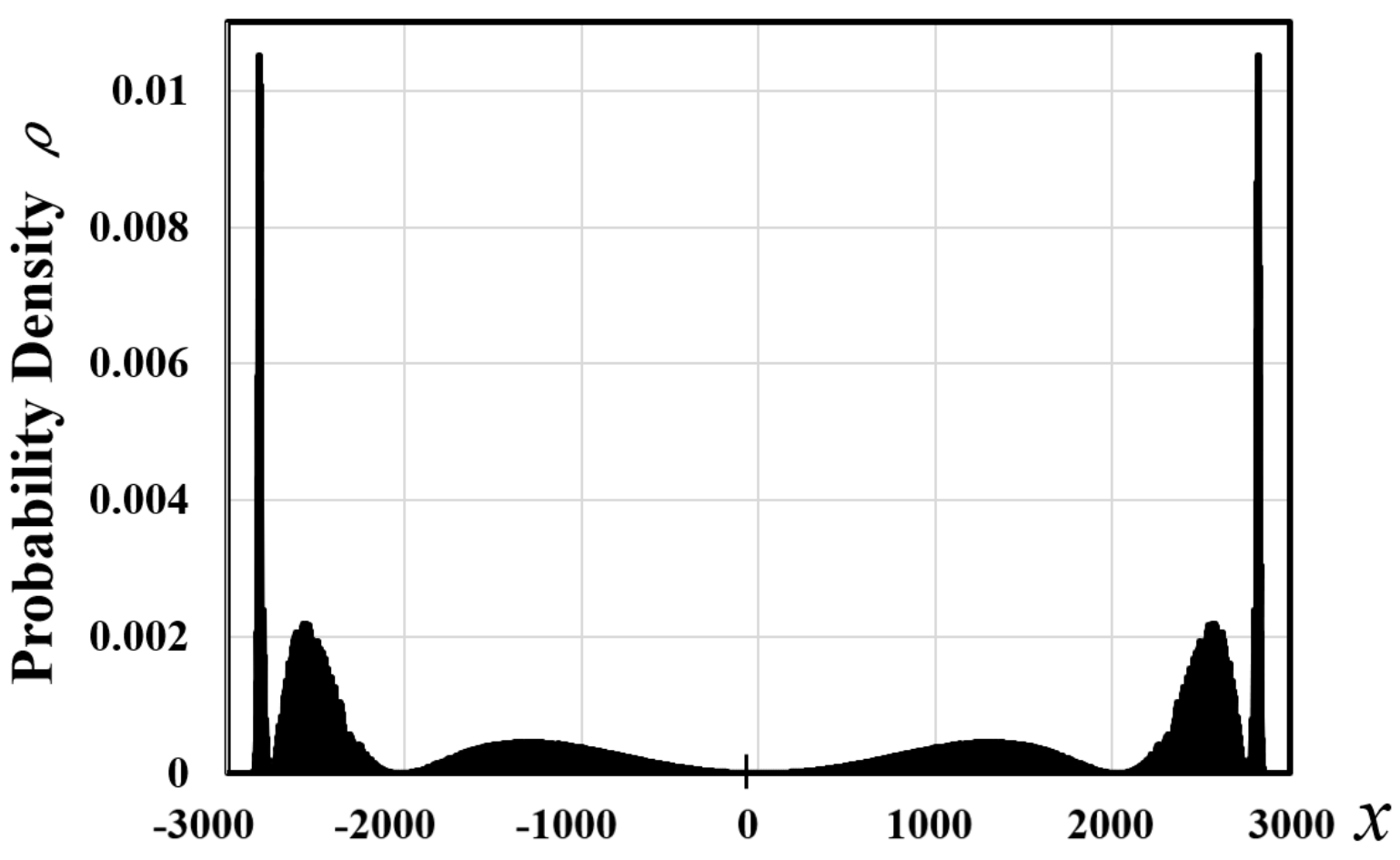}
\hskip 0.05\hsize
\includegraphics[width=0.29\hsize]{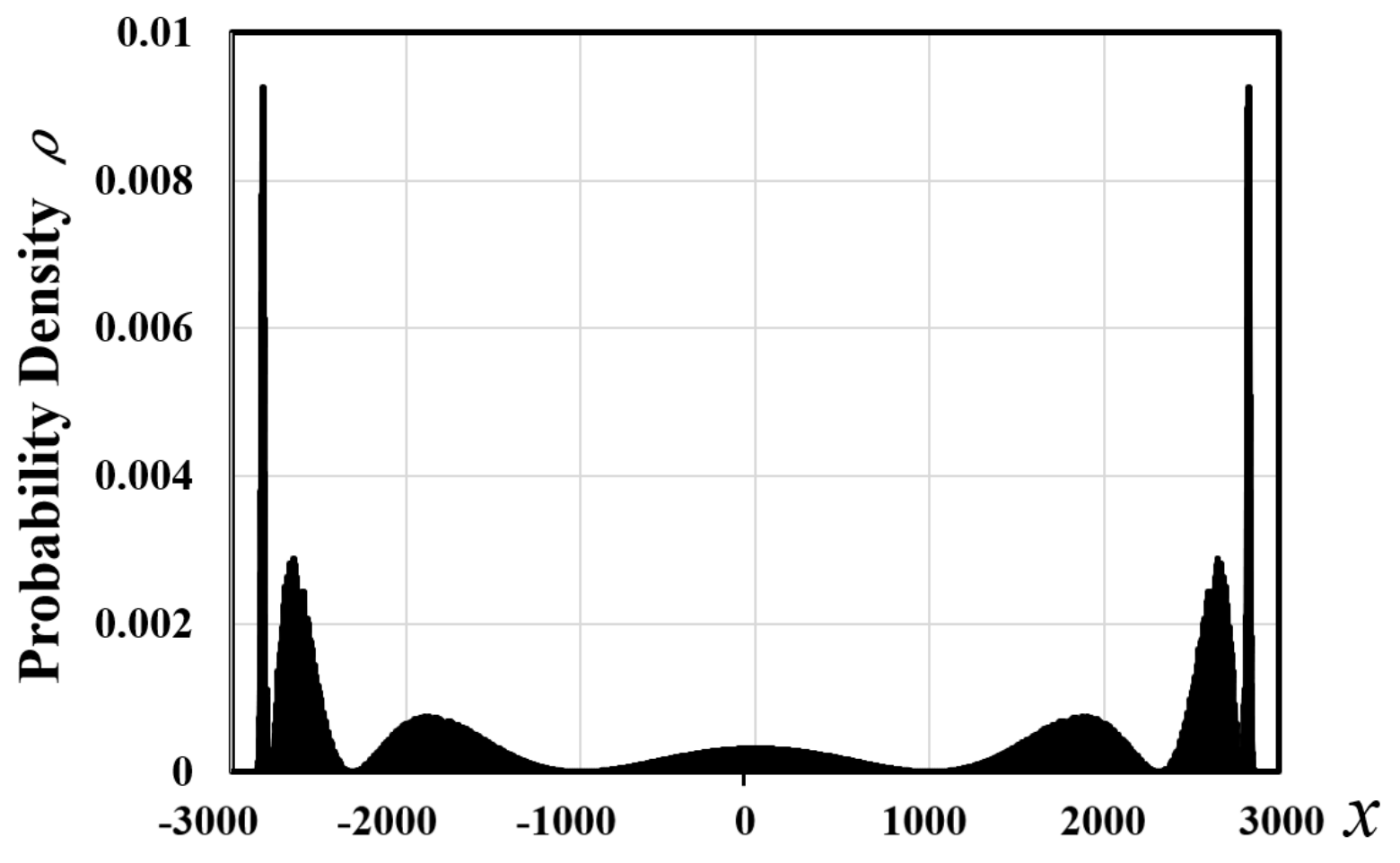}
\hskip 0.05\hsize
\includegraphics[width=0.29\hsize]{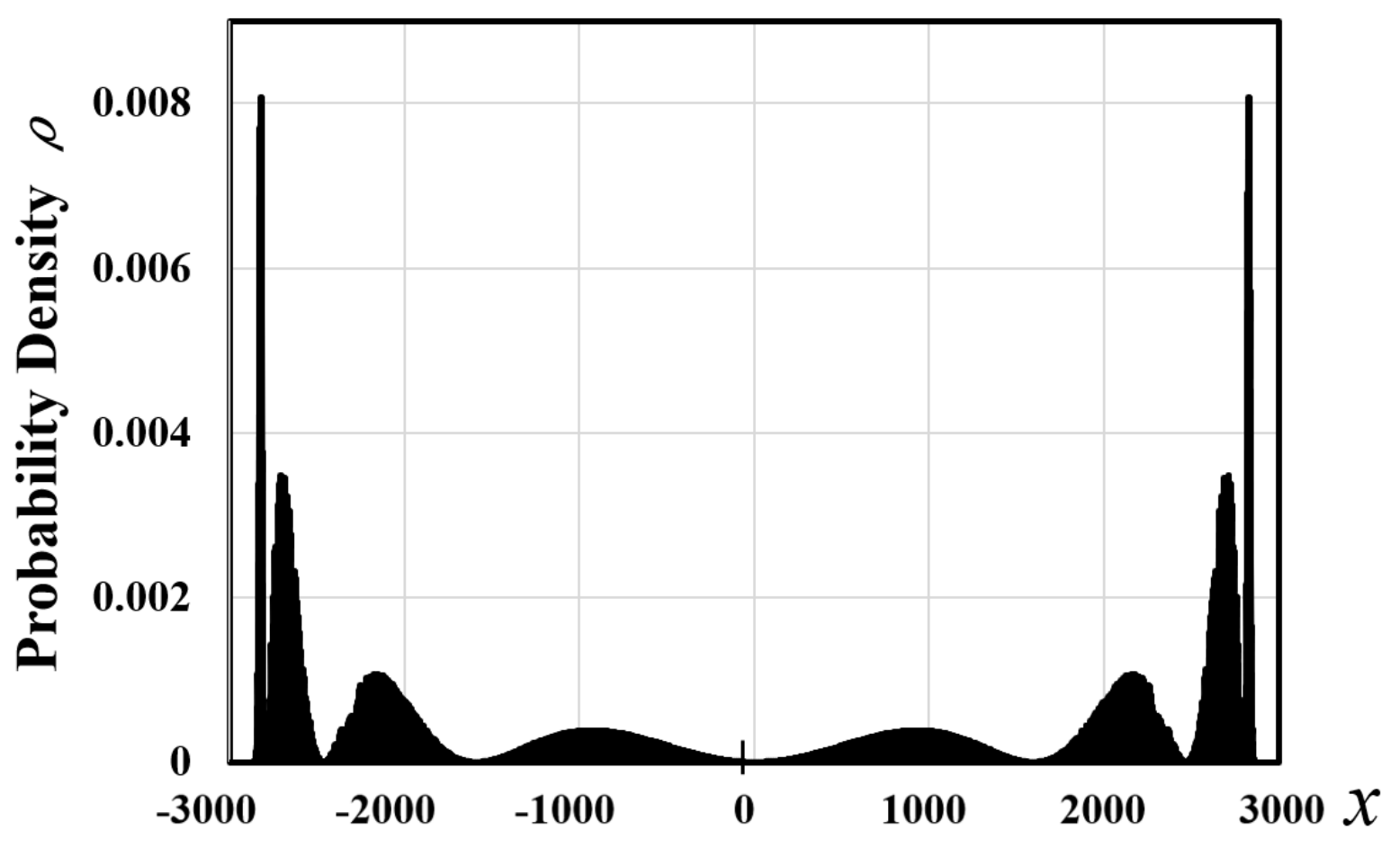}

(d)
\hskip 0.29\hsize
(e)
\hskip 0.29\hsize
(f)

\smallskip

\end{center}

\caption{
Double slit model in the QW. 
(a)$\ell=2$, (b)$\ell=3$, (c)$\ell=4$,
(d)$\ell=5$, (e)$\ell=6$, (f)$\ell=7$.
}\label{fg:WslitAS}
\end{figure}

\begin{theorem}\label{thm:Wslit}
For $\ell  \in \mathbb{Z}_{>0}$, we consider the Hadamard walk $X_t$ whose initial state starts from the two sites $-\ell $ and $\ell $:
$$
\Psi(0,x)=\frac{\varphi}{\sqrt{2}}\delta_{x,-\ell }+\frac{\varphi}{\sqrt{2}}\delta_{x,\ell},
$$
where $\varphi=\left[\begin{matrix}\alpha\\ \beta\end{matrix}\right]\in\mathbb{C}^{2}$ with $|\alpha|^{2}+|\beta|^{2}=1$. Then
$$
\frac{X_{t}}{t}\Rightarrow Z_{\ell }\quad (t \to\infty).
$$
Here the density function $f^{(\ell )}(x)$ of $Z_\ell $ is equal to
\[
f^{(\ell )}(x)=2\,T_{\ell}\!\Bigl(-\frac{x}{\sqrt{1-x^{2}}}\Bigr)^{2}\,f_{K}(x),
\]
where $T_{\ell}(\cdot)$ is the Chebyshev polynomial of the first kind.
\end{theorem}

\begin{proof}
We use the GJS method to prove this theorem as the extension of the proof of Proposition \ref{pr:KonnoDist} in Appendix B.

The Fourier transform $\mathcal{F}\colon (\CC^2)^{\ZZ} \to (\CC^2)^{\KK}$ is defined by
$\widehat{\Psi}(t,k) := \sum_{x\in\ZZ} \Psi(t,x)\, e^{-ikx}$ for any $\Psi(t,\cdot) \in (\CC^2)^{\ZZ}$ and $k \in \KK$.
Then $\widehat{\Psi}(t,k) = U(k)^{t}\,\widehat{\Psi}(0,k)$, where
$U(k) = \left[\begin{matrix} e^{ik} & 0 \\ 0 & e^{-ik} \end{matrix}\right] U$,
and $\widehat{\Psi}(0,k)$ is described by

\[
\widehat{\Psi}(0,k)
=\sum_{x\in\mathbb{Z}} e^{-ikx}\Psi(0,x)
=\left(  \frac{e^{-ik\ell}}{\sqrt{2}} + \frac{e^{ik\ell}}{\sqrt{2}}  \right)\varphi
=\sqrt{2}\cos(\ell k)\,\varphi.
\]
Let $Y:=Z_\ell$ be the weak limit of $X_t/t$ as $t \to \infty$, which exists by Theorem \ref{th:GJS}.
Let $k(y)\in(0,\pi)$ be the unique solution in $(0,\pi)$ of
\[
y=h_{1}(k(y))
=-\frac{\cos (k(y))}{\sqrt{1+\cos^{2}(k(y))}},
\]
i.e., $k(y)=\arccos\bigl(-y/\sqrt{1-y^{2}}\bigr)$. Then the distribution function of $Y$ is given by
\[
P(Y\le y)=
\int_{-k(y)}^{k(y)}
2\cos^{2}(\ell k)\,p_{1}^{(1)}(k)\,\frac{dk}{2\pi}
+\left(
\int_{-\pi}^{-\pi+k(y)}+\int_{\pi-k(y)}^{\pi}
\right)
2\cos^{2}(\ell k)\,p_{2}^{(1)}(k)\,\frac{dk}{2\pi},
\]

where $p_{j}^{(1)}(k)=\bigl|(v_{j}(k),\varphi)\bigr|^{2}$.

By differentiation, the corresponding density function becomes

\begin{eqnarray*}
f^{(\ell)}(y)&=&\frac{d}{dy}P(Y\le y)\\
&=&\frac{1}{2\pi}\,2\cos^{2}\bigl(\ell k(y)\bigr)
\Bigl\{
p_{1}^{(1)}(k(y))+p_{1}^{(1)}(-k(y))
+p_{2}^{(1)}(-\pi+k(y))+p_{2}^{(1)}(\pi-k(y))
\Bigr\}
\frac{dk(y)}{dy}.
\end{eqnarray*}

Here we use
\[
\cos ( k(y) )=-\frac{y}{\sqrt{1-y^{2}}},\qquad
\frac{dk(y)}{dy}=\frac{1}{(1-y^{2})\sqrt{1-2y^{2}}},
\]
and substitute the quantum state (quantum coin) dependent weighted density for the single-point initial state:
\[
f_{K}(y)=
\Bigl\{
1-(|\alpha|^{2}-|\beta|^{2}
+\alpha\overline{\beta}+\overline{\alpha}\beta)\,y
\Bigr\}
\frac{1}{\pi(1-y^{2})\sqrt{1-2y^{2}}}.
\]
Then we finally obtain
\begin{equation}
f^{(\ell)}(y)
=2\cos^{2}\!\bigl(\ell \,k(y)\bigr)\,f_{K}(y),
\qquad |y|<\frac{1}{\sqrt{2}}.
\label{eq:fKcos}
\end{equation}
Normalization follows from
\[
\int_{-\pi}^{\pi}\!\Bigl(\lvert
(v_1(k),\widehat{\Psi}(0,k))_k\rvert^2
+\lvert( v_2(k),\widehat{\Psi}(0,k))_k\rvert^2\Bigr)\,\frac{dk}{2\pi} = \int_{-\pi}^{\pi}2\cos^{2}(\ell k)\,\frac{dk}{2\pi}=1 \quad(\ell\ge1),
\]
where we have used $p_{1}^{(1)}(k)+p_{2}^{(1)}(k)=1$.
\end{proof}

\begin{corollary}\label{cor:4.2}
The following formulas are obtained by expanding the interference factors from (\ref{eq:fKcos}) by means of Chebyshev polynomials for $\ell=1,\ldots,7$.
\begin{enumerate}

\item $\displaystyle{
f^{(1)}(x) =\frac{2 x^2}{1-x^2} f_K(x),
}$

\item $\displaystyle{
f^{(2)}(x) =\frac{2(1-3 x^2)^2}{(1-x^2)^2} f_K(x),
}$

\item $\displaystyle{
f^{(3)}(x) =\frac{2x^2(3-7 x^2)^2}{(1-x^2)^3} f_K(x),
}$

\item $\displaystyle{
f^{(4)}(x) =\frac{2(17 x^4-10 x^2+1)^2}{(1-x^2)^4} f_K(x),
}$

\item $\displaystyle{
f^{(5)}(x) =\frac{2x^2(41 x^4-30 x^2+5)^2}{(1-x^2)^5} f_K(x),
}$

\item $\displaystyle{
f^{(6)}(x) =\frac{2(99 x^6-87 x^4+21 x^2-1)^2}{(1-x^2)^6} f_K(x),
}$

\item $\displaystyle{
f^{(7)}(x) =\frac{2x^2(239 x^6 - 245 x^4 + 77 x^2 - 7)^2}{(1-x^2)^7} f_K(x).
}$
\end{enumerate}
\end{corollary}

The shapes of the profile are illustrated in Figure \ref{fg:fell}.

\bigskip
\begin{figure}
\begin{center}

\includegraphics[width=0.29\hsize]{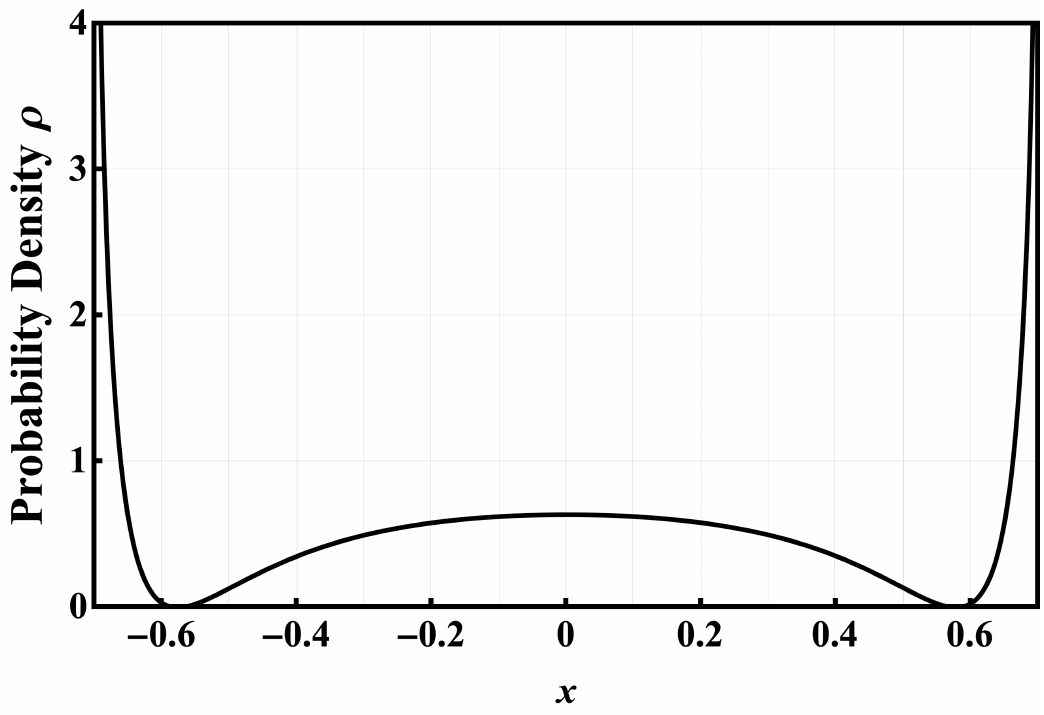}
\hskip 0.05\hsize
\includegraphics[width=0.29\hsize]{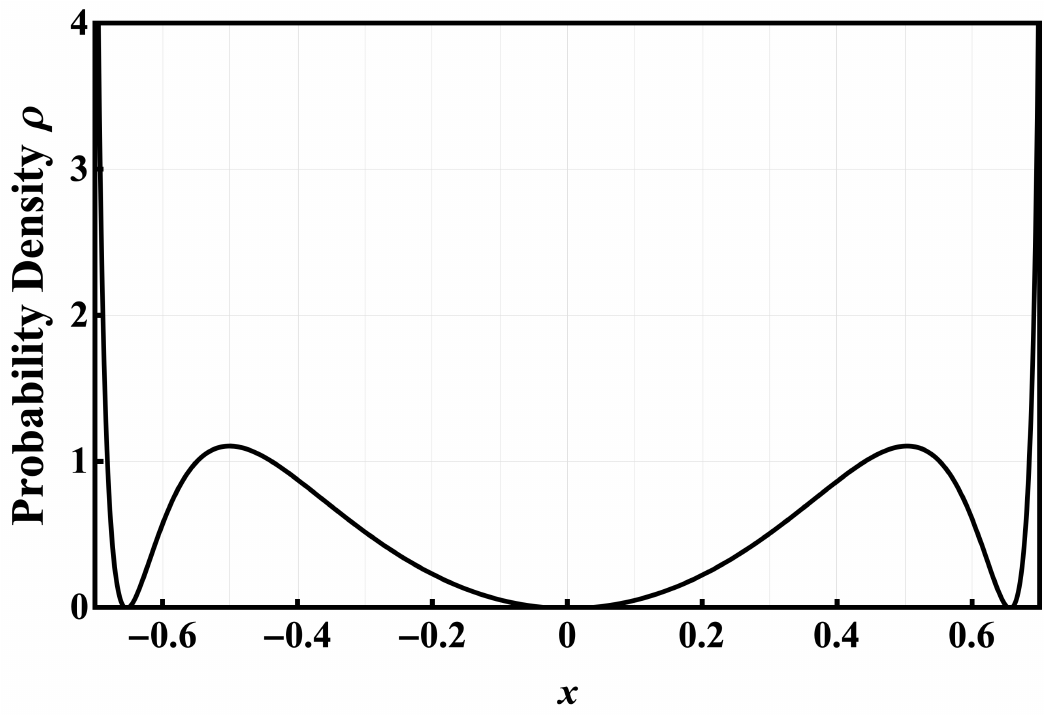}
\hskip 0.05\hsize
\includegraphics[width=0.29\hsize]{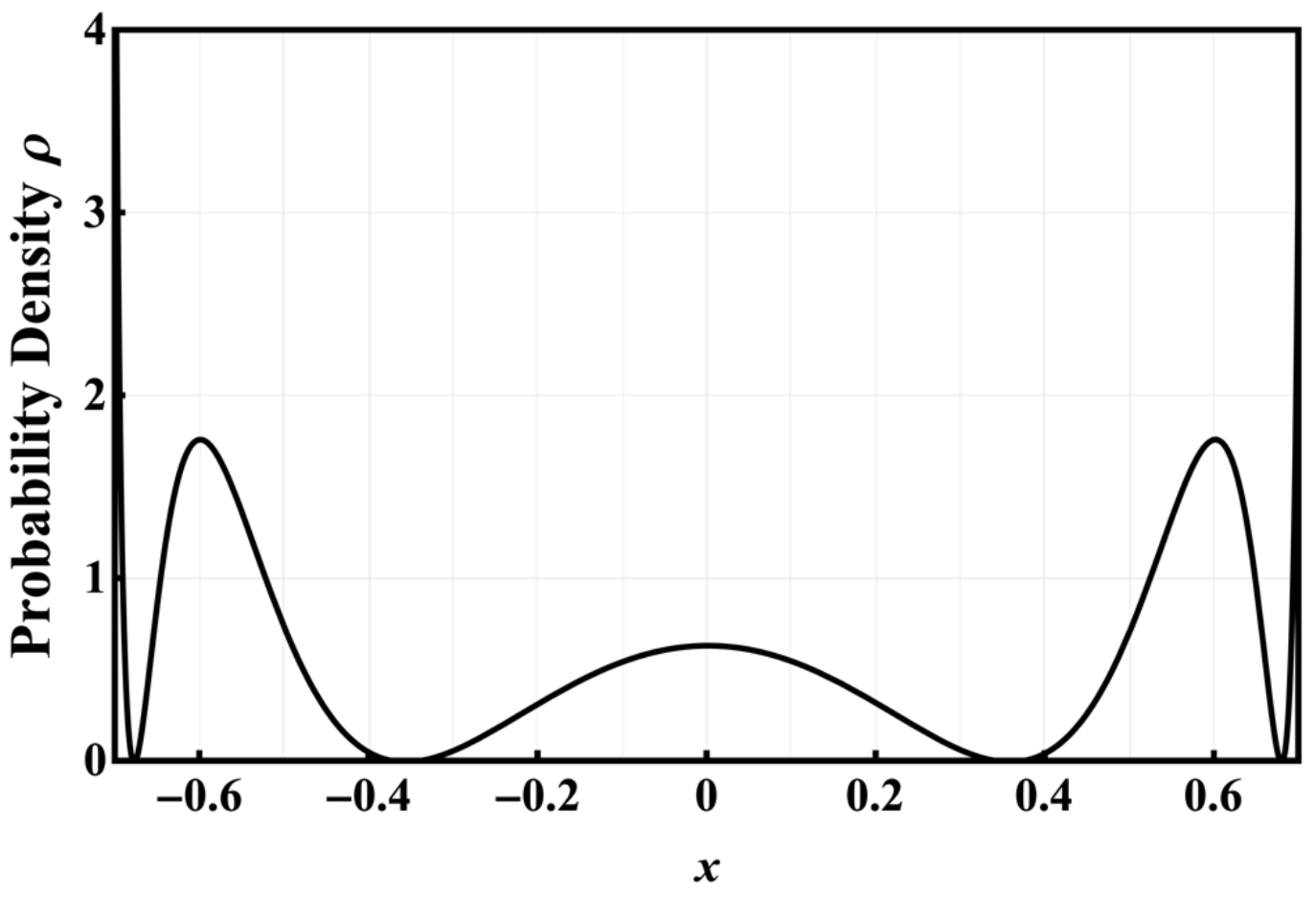}

(a)
\hskip 0.29\hsize
(b)
\hskip 0.29\hsize
(c)

\includegraphics[width=0.29\hsize]{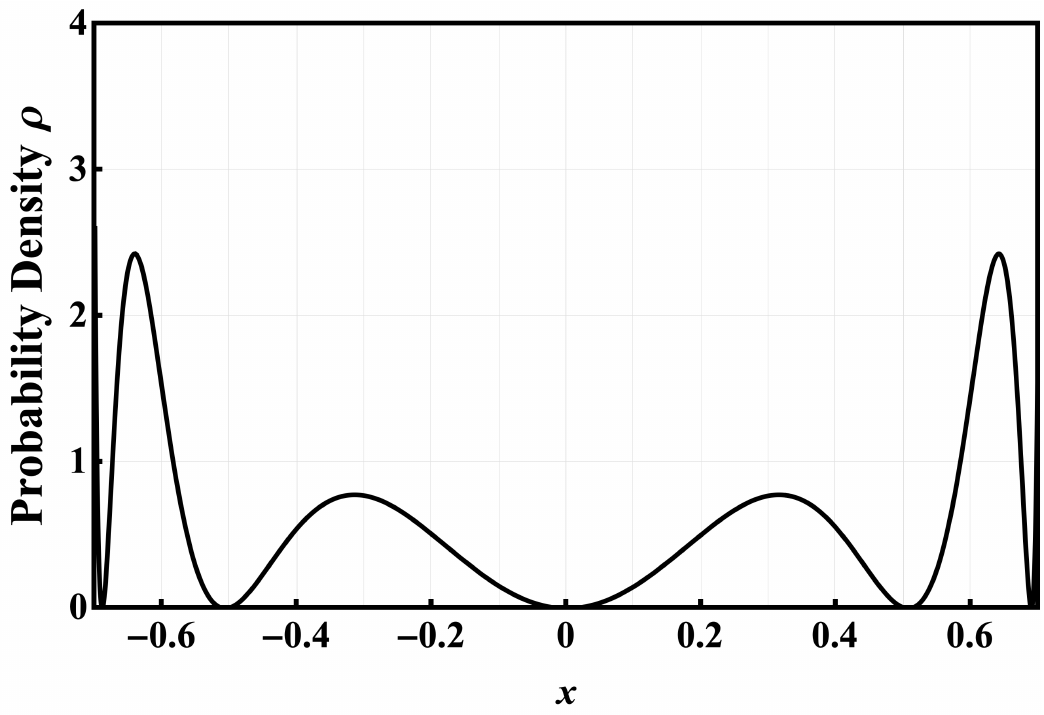}
\hskip 0.05\hsize
\includegraphics[width=0.29\hsize]{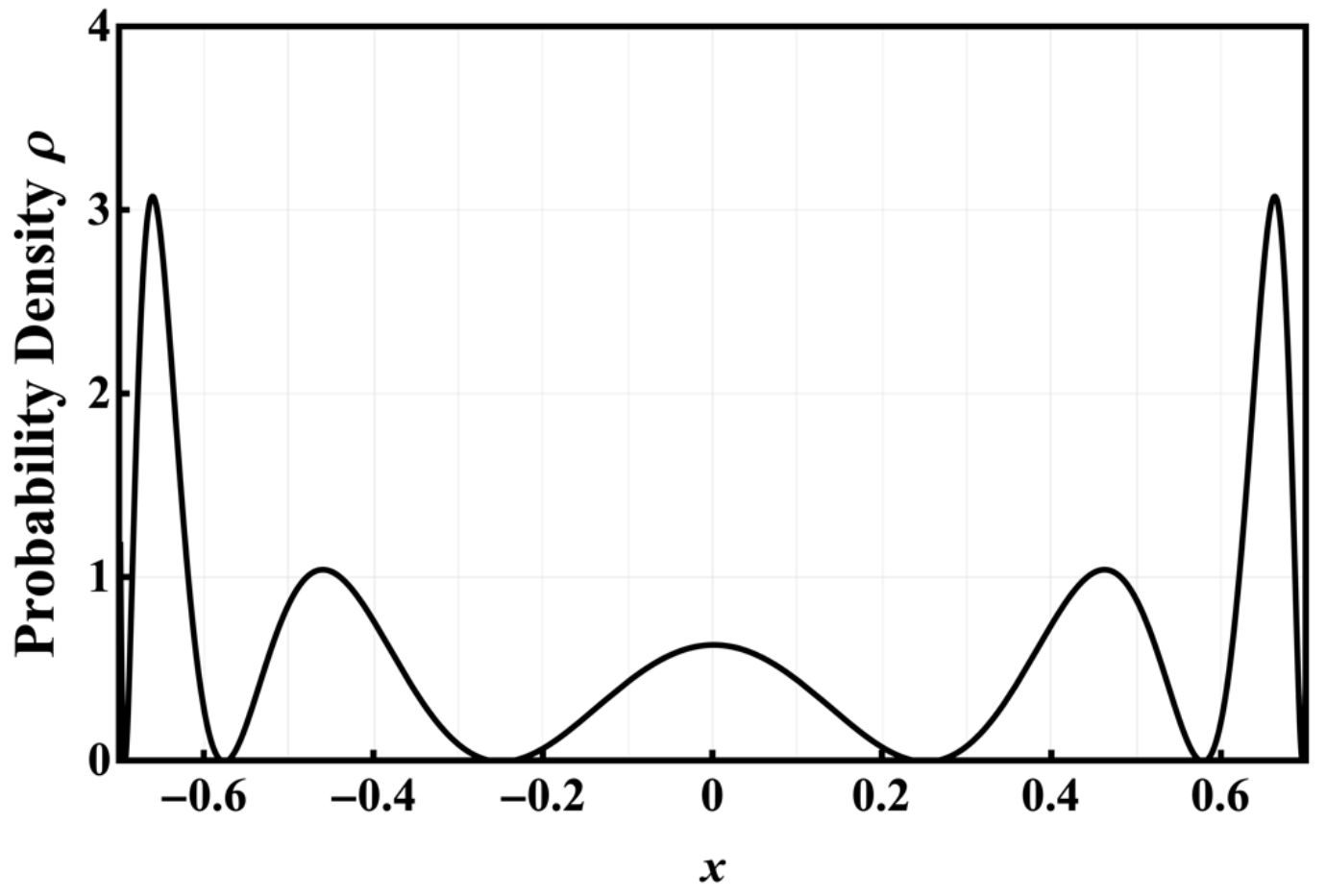}
\hskip 0.05\hsize
\includegraphics[width=0.29\hsize]{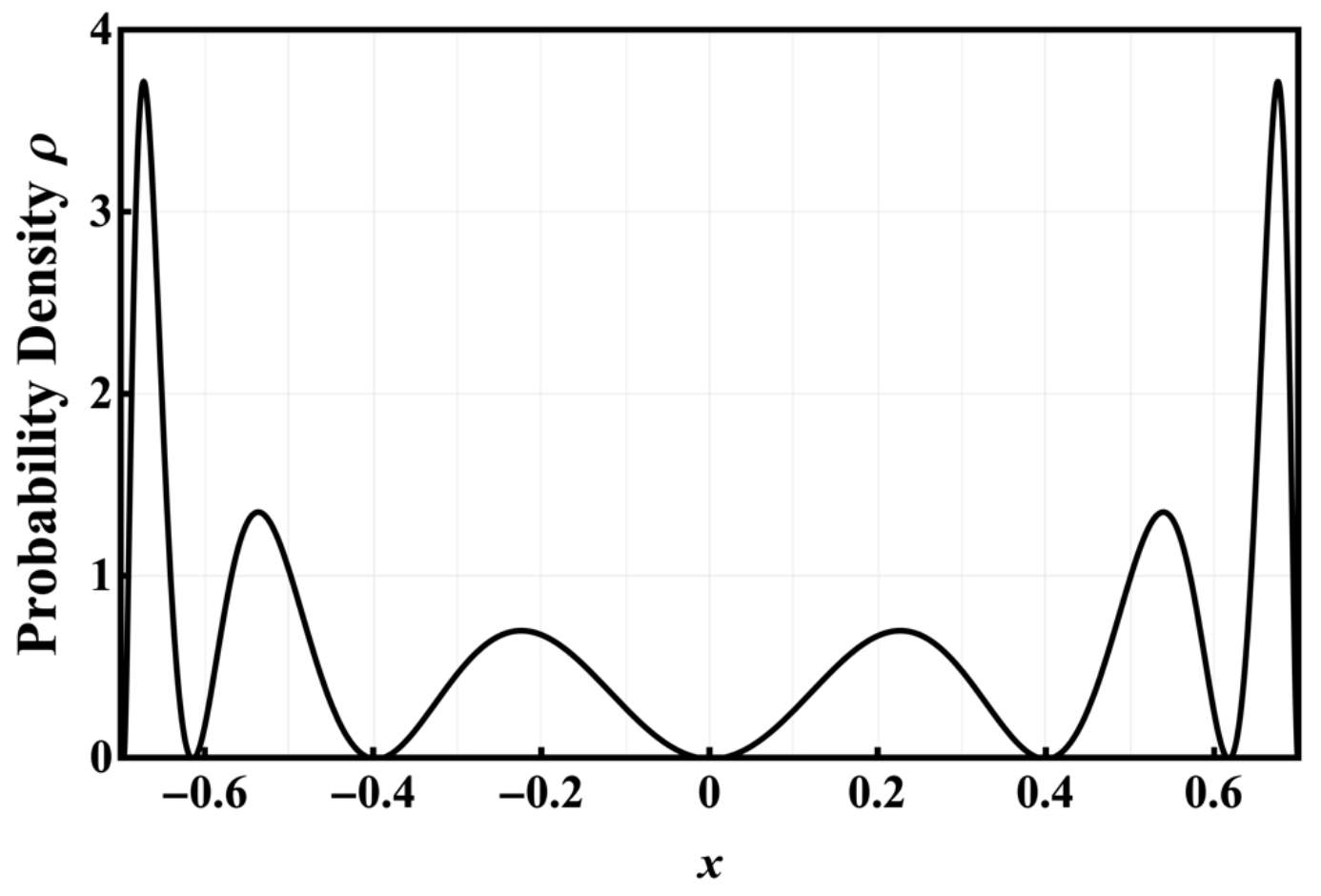}

(d)
\hskip 0.29\hsize
(e)
\hskip 0.29\hsize
(f)

\smallskip

\end{center}

\caption{
Double slit model in the QW. 
(a)$\ell=2$, (b)$\ell=3$, (c)$\ell=4$,
(d)$\ell=5$, (e)$\ell=6$, (f)$\ell=7$.
}\label{fg:fell}
\end{figure}

\begin{remark}\label{rk:4.3}
{\rm{
We demonstrate that the QW provides an interference phenomenon by considering quantum states stemming from two points.
Depending on the distance $2\ell$ between the double slit, the profile of the  probability density $\rho(t,x)$ exhibits a different aspect as in Figure \ref{fg:WslitAS}.

We also theoretically obtain the weak limit expression of the asymptotic behavior by means of the GJS-method, as in Theorem \ref{thm:Wslit}.
The typical cases are listed in Corollary \ref{cor:4.2}, which are illustrated in Figure \ref{fg:fell}.

Figure \ref{fg:fell} describes well the numerical evaluations in Figure \ref{fg:WslitAS}.
In other words, we demonstrate that the QW model also provides the interference, numerically and theoretically besides the particle-like behavior in Section 3.

It is remarkable that as far as we know, no previous reports have documented instances of QW exhibiting interference effects. We believe this is an extremely significant discovery.

\bigskip

Before finishing this remark, we give comments on its mathematical meaning.
As in proof of Theorem \ref{thm:Wslit}, the profile $f^{(\ell)}(x)$ comes from the Chebyshev polynomials $T_\ell$,
i.e., $T_\ell(x) = \displaystyle{ \sum_{i=0}^\ell a_{\ell, i} x^i}$ is generated by $T_\ell(\cos u) = \cos(\ell u)$.

Intensity is fixed by the interference condition, such as $\cos(\ell k)$.

Behind the cosine function, we may find the cyclotomic field $\QQ[\mathrm{e}^{2\pi i/\ell}]$, i.e., $T_\ell\left(\cos\frac{u}{\ell}\right)=\cos u$ in number theory, which plays an essential role in near-field interference phenomena and the Talbot effect \cite{BK, MO} from a mathematical viewpoint \cite{LW}.

}}
\end{remark}

\section{Conclusion and Discussion}\label{sec:5}

As we demonstrated the computations of the QW model, a quantum mechanical model, from a quantum mechanical perspective, we first comment on its relationship with Feynman's path integral.
As mentioned in Remark \ref{rmk:PIvsQW}, the picture in (\ref{eq:PsiKPsi}) plays the essential role in both the path integral and the QW model.
The QW is a simple but rigorous mathematical model to express the relation (\ref{eq:PsiKPsi}).
Thus, we should re-evaluate the behaviors of the QW model from the quantum mechanical perspective.

\bigskip

As mentioned in Introduction, although there are many quantum models, few demonstrate classical correspondence through a parameter as far as we know.
In QWs, we recover classical particles by taking the limit of the quasi-inner scattering parameter $s$, $s\to0$.
Although the QW model is not directly related to the Planck constant $\hbar$, we can interpret $s$ as a parameter bridging particle-like and wave-like properties.

\begin{center}
\begin{tabular}{|c|c|}
\hline
small $s$ & large $s$ \\
\hline
classical particle-like & wave-like \\
\hline
\hline
velocity 
to 1 [space step/time step] &velocity  to zero \\
\hline
  \end{tabular}
\smallskip
\end{center}

The fact that the velocity of the wavefront (the peak of the wave profile) is less than the velocity of information propagation is noteworthy.
This reminds us of the pilot wave of Bohm and the relation between the group velocity and phase velocity, although, the QW model does not provide a deterministic picture or hidden variables.

\bigskip

Furthermore, it is interesting that the wavefront in the Konno distribution is defined by the divergence point, even though its integration is finite.
This is reminiscent of caustic phenomena in optics, which were studied well by Berry and Upstill \cite{BU}.
The famous caustic phenomenon is a rainbow in the sky. The mechanism behind it was revealed by Descartes and Newton.
It is well known that Airy resolved the divergence of the rainbow, or caustics, in the framework of geometric optics by investigating the wave effect \cite{BU}
as we redisplay the Fig. 3.2 in \cite{BU}, which provides a comparison of the Airy function and the observation of diffraction intensity at a fold caustic, as in Figure \ref{fg:BU}.
Figure \ref{fg:ST_dist} closely resembles the intensity of the fold caustic in Figure \ref{fg:BU} and the intensity distribution of a rainbow derived by Airy based on wave theory. 
Additionally, the divergence of the Konno distribution exhibits the same behavior as a rainbow or caustic \cite{A,JR}. 
The Airy function refined the crests of geometrical optics with the wavy effect, as Sunada and Tate did \cite{ST} cf. Remark \ref{rmk:SunadaTate}.

As Schulman investigated the treatment of the caustics in the path integral method \cite{Schul}, the profile of the probability density in the QW model is a very interesting problem itself.

\begin{figure}
\begin{center}

\includegraphics[width=0.7\hsize]{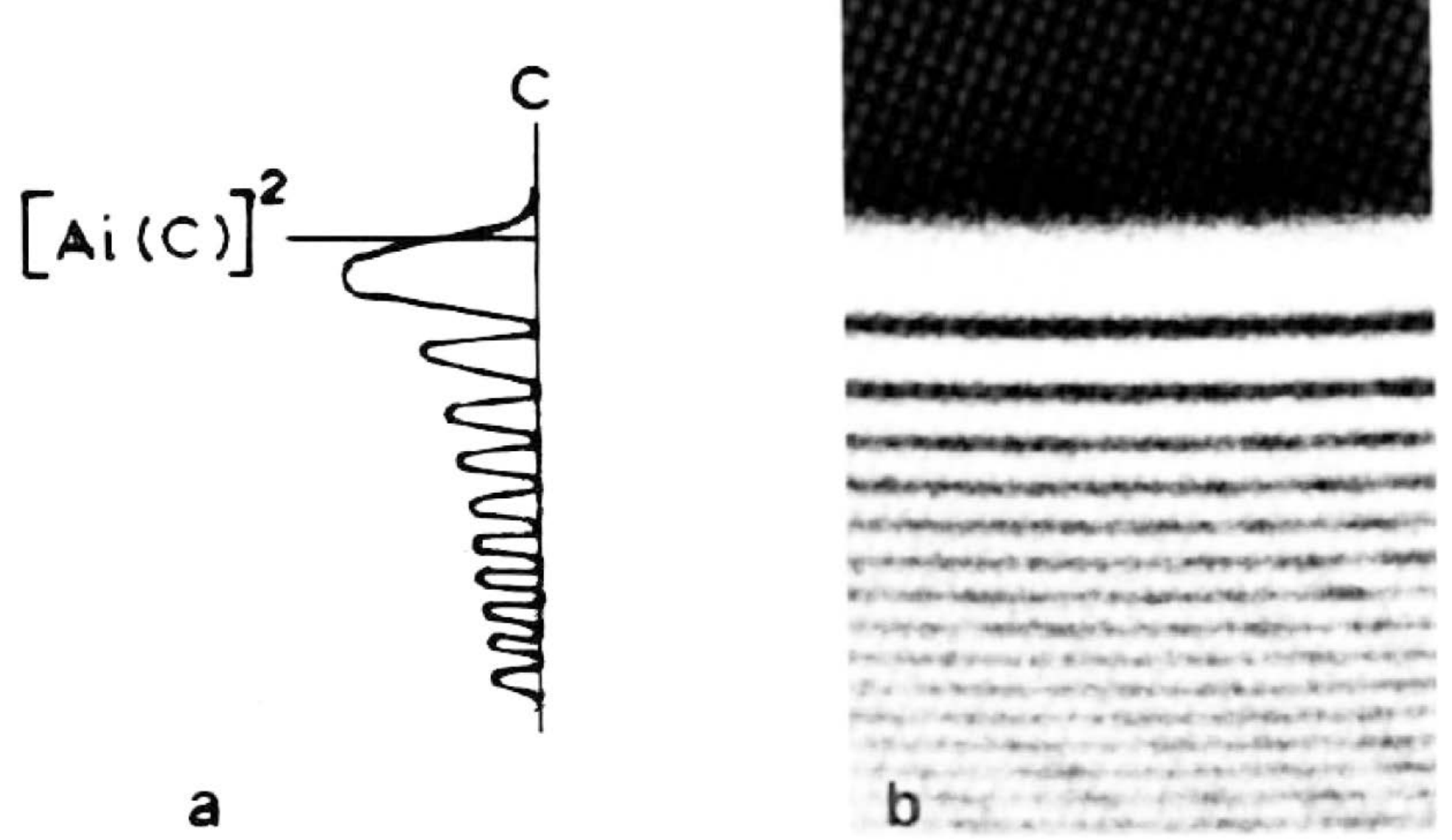}
\

\smallskip

\end{center}

\caption{
Fig. 3.2 in \cite{BU}:
(a) Graph of the fold diffraction catastrophe intensity (square of Airy function); (b) observation of diffraction intensity at a fold caustic.
}
\label{fg:BU}
\end{figure}

\bigskip

Further, although the interference effect has never been reported in the QW in the long term, in this paper, we also demonstrated our observation of double-slit interference in the QW in Section \ref{sec:4} with the rigorous proof as in Theorem \ref{thm:Wslit}.
We believe this is an extremely significant discovery.
This phenomenon shows both naturalness and surprise and can be interpreted as the Fraunhofer type, far-field diffraction as mentioned in Section \ref{sec:4}.

\bigskip

Thus, the QW model is a natural mathematical toy model for exploring the correspondence between classical and quantum mechanics.
The QW model contains the quasi-inner scattering parameter $s$, which connects classical and quantum behavior to the path integral picture.
This correspondence may be associated with the well-known caustic phenomenon and can be investigated more precisely using singularity theory and singular perturbation method \cite{BU,Schul}.
Furthermore, the QW model exhibits the interference phenomenon well.
In other words, as demonstrated in this paper, the QW model exhibits wave-particle complementarity.
From a mathematical perspective, the QW model is essential for understanding the distinction between quantum mechanics and classical mechanics, as well as the role of the path integral in that distinction.

This also has the potential to solve various fundamental problems that were impossible to resolve during the formative period of quantum mechanics due to a lack of mathematical advances.
We believe that now is the time to understand and achieve these goals from a contemporary perspective.

\bigskip
Before finishing this paper, we give some more comments on the QW model from a mathematical viewpoint.

As a particle picture, we mentioned the caustics above.
The interference in Section \ref{sec:4} also reminds us of the wavy caustic problem, i.e., the rational Talbot effect in near-field optical phenomena, which Berry and Klein studied well \cite{BK}.
The rational Talbot effect is closely related to the theory of quadratic residues in number theory, and number-theoretic properties describe wave phenomena in nature.
Furthermore, we observe wave-particle complementarity in the rational Talbot effect, which manifests as the quadratic reciprocity law in number theory through the most straightforward number-theoretic path integral, the Gauss sum \cite{MO}.
Our interference problem is also related to the cyclotomic field as mentioned in Remark \ref{rk:4.3} \cite{LW}.
In other words, behind the interference phenomenon in Section 4, there might be number theoretical background.

Recently, the QW model was extended to quantum states on a graph, which are related to zeta functions in number theory \cite{AKOS}.
A century has passed since the establishment of quantum mechanics. Physics and mathematics have both advanced in this time, leading to the surprising development of remarkable theory of the QW model.
We believe that incorporating these modern mathematical perspectives into the QW model allows us to approach the quantum world from a contemporary viewpoint, tackling it as a rigorous mathematical problem.

\appendix
\section{Feynman's checkerboard Dirac model: Problem~2-6  in \cite{FHS}}

\begin{figure}
\begin{center}

\includegraphics[width=0.4\hsize]{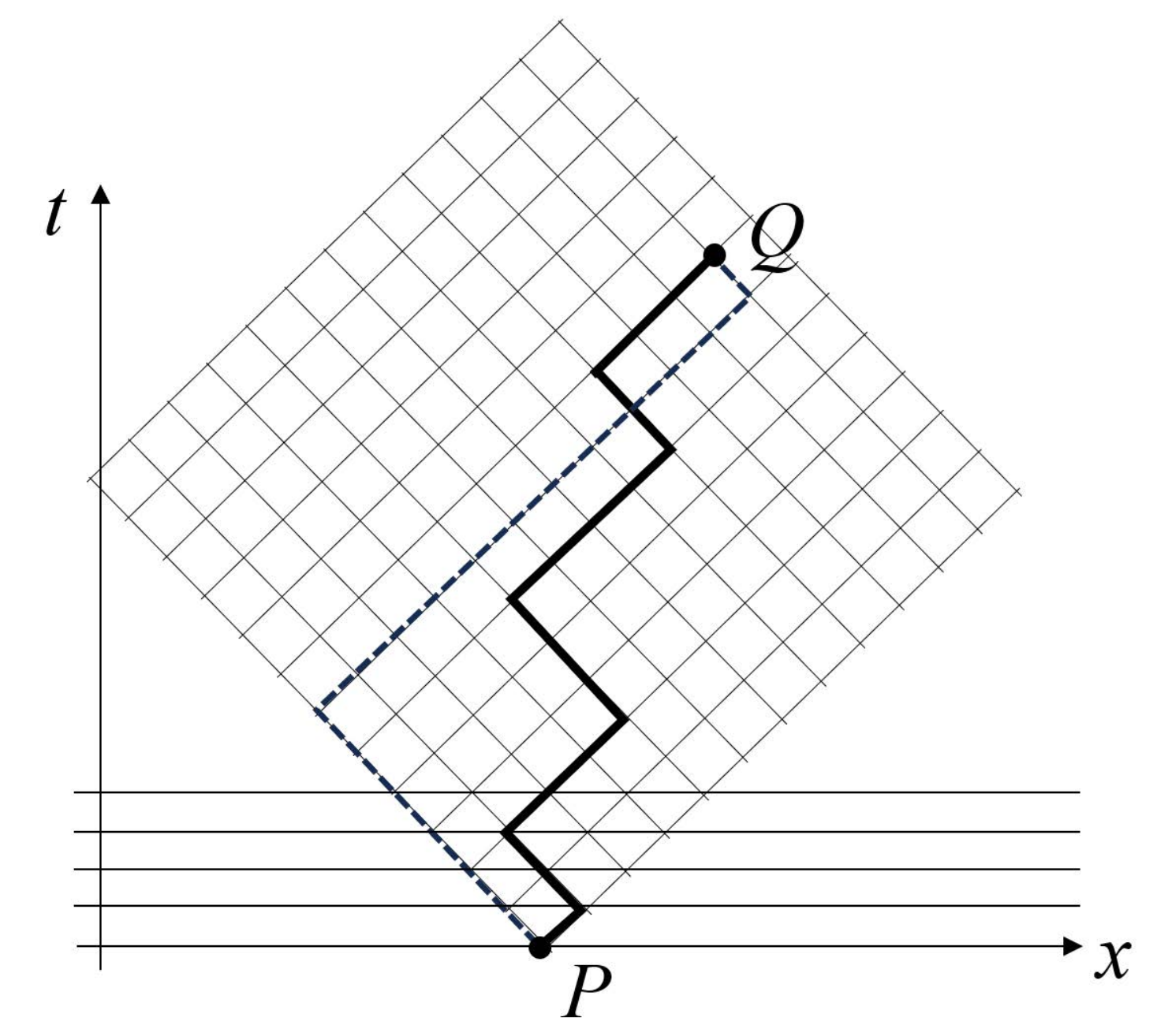}
\

\smallskip

\end{center}

\caption{
Zigzag diagram
}\label{fg:Zigzag}
\end{figure}

Feynman considered the class of functionals for which path integrals can be defined.
While the general kernel is expressed as
\[
K(b,a)=\int_{a}^{b} e^{i S[x(t)] / \hbar}\, \mathcal{D}x(t),
\]
he also introduced a distinct type of functional arising from a one-dimensional relativistic model.

\medskip

Suppose that a particle moves only forward or backward at the velocity of light.
Then, in the $(x,t)$ plane, all trajectories oscillate with slopes $\pm 45^\circ$, as schematically illustrated in Feynman's original figure as in Figure \ref{fg:Zigzag}.
Time is divided into small equal steps of length $\varepsilon$, and reversals of direction
are allowed only at the boundaries of these steps, i.e., at $t = t_a + n\varepsilon$ with $n \in \mathbb{Z}$.

\medskip

For this relativistic case, Feynman defined the amplitude of a given path not by the exponential weight
of the action but by the number of corners (reversals) along the trajectory.
Specifically, if a trajectory contains $R$ reversals, the amplitude is given by
\[
\phi = (i\varepsilon)^R.
\]
Consequently, the total transition amplitude (kernel) between spacetime points $a$ and $b$
is obtained by summing over all possible paths,
\[
K(b,a) = \sum_{R} N(R) (i\varepsilon)^R,
\]
where $N(R)$ denotes the number of distinct trajectories containing $R$ corners.

\medskip

This formulation 
forms the conceptual bridge between the path-integral formalism and the unitary dynamics on discrete spacetime lattices---precisely the foundation of the modern QW model.


\section{Proof of Proposition \ref{pr:KonnoDist}}

We compute explicitly in the case of the Hadamard walk.
The Fourier representation of the one-step evolution operator is
\[
U(k)=\frac{1}{\sqrt{2}}
\begin{bmatrix}
e^{ik}   &e^{ik} \\
e^{-ik}  &-e^{-ik}
\end{bmatrix}.
\]
Let $I=I(k)=1+\cos^2(k)$. Then the eigenvalues of $U(k)$ are given by
\[
\lambda_{1}(k) = \frac{1}{\sqrt{2}} \left( \sqrt{I} + i \sin k  \right) , \quad
\lambda_{2}(k) = \frac{1}{\sqrt{2}} \left( -\sqrt{I} + i \sin k  \right) ,
\]
and for $D=i\frac{d}{dk}$, we obtain
\[
h_{j}(k) = D\lambda_{j}(k)/\lambda_{j}(k), \quad  h_{1}(k) = -\frac{\cos k}{\sqrt{I}}, \quad h_{2}(k) = \frac{\cos k}{\sqrt{I}}.
\]

The eigenvectors can be chosen as 
\[
v_{1}(k) = \sqrt{ \frac{\sqrt{I}+\cos k}{2\sqrt{I}} }
\begin{bmatrix}
e^{ik}    \\
\sqrt{I} - \cos k 
\end{bmatrix}, 
\quad
v_{2}(k) = \sqrt{ \frac{\sqrt{I}-\cos k}{2\sqrt{I}} }
\begin{bmatrix}
e^{ik}    \\
-\sqrt{I} - \cos k 
\end{bmatrix},
\]
so that $(v_{j}(k), v_{j}(k))_k = 1$ for $j=1,2$.
Let
$$
p_{j}(k) = \lvert(v_{j}(k),\widehat{\Psi}(0,k))_k\rvert^2 \quad (j=1,2).
$$
We consider the initial state localized at the origin with the quantum state $\varphi = \displaystyle{\left[\begin{matrix}\alpha\\ \beta\end{matrix}\right]}$ satisfying $|\alpha|^{2}+|\beta|^{2}=1$; that is, 
$\Psi(0,x)= \varphi \delta_{x,0}$.
Hence 
\[
p_{j}(k) = \lvert( v_{j}(k), \varphi)_k\rvert^2 \quad (j=1,2).
\]
Here we also note that 
\[
p_{1}(k) + p_{2}(k) = 1 \quad  (k \in [-\pi,\pi) ).
\]
More generally, even without localization at the origin, the same arguments below apply provided that
$$
\int_{-\pi}^{\pi}\!\Bigl(\lvert( v_1(k), \widehat{\Psi}(0,k))_{k}\rvert^2
+\lvert(v_2(k),\widehat{\Psi}(0,k))_{k}\rvert^2\Bigr)\,\frac{dk}{2\pi} = 1.
$$
By Theorem \ref{th:GJS}, 
\[
P(Y \le y) = \int_{ \{ k \in [-\pi,\pi): h_{1}(k)\le y \} } p_{1}(k) \frac{dk}{2\pi} + \int_{ \{ k \in [-\pi,\pi): h_{2}(k)\le y \} } p_{2}(k) \frac{dk}{2\pi}.
\]
Let $k=k(y) \in (0,\pi) $ denote the unique solution of $h_{1}(k)=-\cos k / \sqrt{I}=y$ for $-1/\sqrt{2}<y<1/\sqrt{2}$; explicitly, $k(y)=\arccos\bigl(-y/\sqrt{1-y^{2}}\bigr)$.

Then
\[
P(Y \le y) = \int_{-k(y)}^{k(y)} p_{1}(k) \frac{dk}{2\pi} +  \left( \int_{-\pi}^{-\pi+k(y)} + \int _{\pi-k(y)}^{\pi}  \right)  p_{2}(k)\frac{dk}{2\pi}.
\]

Differentiating both sides yields the density $f(y)$:
\begin{align*}
f(y) &= \frac{d}{dy} P(Y \le y) \\
&= \frac{1}{2\pi} 
\left\{
p_{1}(k(y)) + p_{1}(-k(y)) + p_{2}(-\pi+k(y)) +  p_{2}(\pi-k(y)) 
\right\} 
\frac{dk(y)}{dy}.
\end{align*}

We now express the right-hand side explicitly as a function of $y$.
From the defining relation,
\[
h_{1}(k(y)) = - \frac{\cos(k(y))}{ \sqrt{1+\cos^2(k(y))} }=y,
\]
we obtain
\[
\cos(k(y)) = -\frac{y}{\sqrt{1-y^2}},  \quad
\sin(k(y)) = \sqrt{ \frac{1-2y^2}{1-y^2} }.
\]
Differentiating the identity above with respect to $y$ gives
\[
\frac{dk(y)}{dy} = \frac{1}{(1-y^2)\sqrt{1-2y^2}},
\]
and moreover 
\[
I = 1+\cos^2(k(y)) = \frac{1}{1-y^2}.
\]

Using these relations, together with the explicit eigenvectors, we find
\begin{align*}
v_{1}(k(y)) &= \frac{1}{\sqrt{2(1+y)}}
\begin{bmatrix}
-y+i\sqrt{1-2y^2}   \\
1+y 
\end{bmatrix}, \\
v_{1}(-k(y)) &= \frac{1}{\sqrt{2(1+y)}}
\begin{bmatrix}
-y-i\sqrt{1-2y^2}   \\
1+y 
\end{bmatrix},\\
v_{2}(-\pi+k(y)) &= \frac{1}{\sqrt{2(1+y)}}
\begin{bmatrix}
y-i\sqrt{1-2y^2}   \\
-1-y 
\end{bmatrix}, \\
v_{2}(\pi-k(y)) &= \frac{1}{\sqrt{2(1+y)}}
\begin{bmatrix}
y+i\sqrt{1-2y^2}   \\
-1-y 
\end{bmatrix}.
\end{align*}

Combining the above with the definitions of $p_{j}$ and $I$, one obtains the desired density

\begin{align*}
f(y ) &= \Bigl\{1-\bigl(|\alpha|^{2}-|\beta|^{2}+\alpha\overline{\beta}+\overline{\alpha}\beta\bigr)y \Bigr\}
\frac{1}{\pi (1-y^2) \sqrt{1-2y^2}} I_{(-1/\sqrt{2},\,1/\sqrt{2})}(y) \\
&= \Bigl\{1-\bigl(|\alpha|^{2}-|\beta|^{2}+\alpha\overline{\beta}+\overline{\alpha}\beta\bigr)y \Bigr\} f_{K}\!\left(y;\,\frac{1}{\sqrt{2}}\right).
\end{align*}

\section*{Acknowledgments}
The authors thank Jiro Akahori, Iwao Sato, and Yusuke Ide for their critical discussions related to this study.
The second author, S.M., is also grateful to his students, especially,
Nur Athirah Binti Ahmad, Nurfarah Hani Binti Halim, Azumi Koshino, Taiki Sato, and Hayato Tanabe.

The second author, S.M., has been supported by the Grant-in-Aid for Scientific Research (C) of the Japan Society for the Promotion of Science Grant, No.21K03289 and No.26K06824.
The third author, E.S., has been supported by the Grant-in-Aid for Scientific Research (C) of the Japan Society for the Promotion of Science Grant, No.24K06863.

\bigskip

\end{document}